\documentclass[12pt]{article}
\usepackage{amsthm}
\usepackage{amsmath,bm}
\usepackage{enumerate}
\usepackage{amssymb}
\usepackage{latexsym}
\usepackage{amsfonts}
\usepackage{color}
\usepackage{secdot}
\usepackage{graphicx,epsfig}
\usepackage{mathrsfs}
\usepackage{subfigure}
\usepackage{epstopdf}
\usepackage{epsfig}
\usepackage{natbib}
\usepackage{booktabs,tabularx}
\newtheorem{theorem}{Theorem}[section]

\newtheorem{example}{Example}[section]

\newtheorem{remark}{Remark}[section]

\begin{document}
\title{\bf A novel distribution with upside down bathtub shape hazard rate: properties, estimation and applications}

\author{ Tuhin Subhra {\bf Mahato}\thanks {Email address : tuhinsubhra.mahato@gmail.com},~ Subhankar {\bf Dutta}\thanks {Corresponding author : subhankar.dta@gmail.com},~ and Suchandan {\bf Kayal}\thanks{Email address : kayals@nitrkl.ac.in,~suchandan.kayal@gmail.com}}

\date{}
\maketitle 
\noindent {\it Department of Mathematics, National Institute of
	Technology Rourkela, Rourkela-769008, Odisha, India} \\
{}
\begin{center}Abstract
\end{center}
In this communication, we introduce a new statistical model and study its various mathematical properties. The expressions for hazard rate, reversed hazard rate, and odd functions are provided. We explore the asymptotic behaviours of the density and hazard functions of the newly proposed model. Further,  moments, median, quantile, and mode are obtained. The cumulative distribution  and density functions of the general $k$th order statistic are provided. Sufficient conditions, under which the likelihood ratio order between two inverse generalized linear failure rate (IGLFR) distributed random variables holds, are derived. In addition to these results, we introduce several estimates for the parameters of IGLFR distribution. The maximum likelihood and maximum product spacings estimates are proposed. Bayes estimates are calculated with respect to the squared error loss function. Further, asymptotic confidence and Bayesian credible intervals are obtained. To observe performance of the proposed estimates, we carry out a Monte Carlo simulation using $R$ software. Finally, two real life data sets are considered for the purpose of illustration.  

\section{Introduction}
The statistical literature on lifetime models has been enriched with several continuous distributions and  it is still developing. Many distributions, which have been useful in various real-life situations are introduced in past two decades. For example, \cite{abouammoh2009reliability} proposed a new distribution named as generalized inverted exponential (GIE) distribution and estimated its reliability parameter using various estimation techniques. Further, \cite{de2011generalized} introduced and studied a three-parameter generalized inverse Weibull distribution. The authors observed that the failure rate of this distribution is decreasing and unimodal.
\cite{tahir2018inverted} proposed  an inverted model called the inverted Nadarajah-Haghighi distribution. Some properties of this distribution and estimates of  model parameters have been explored by the authors. \cite{basheer2019alpha} introduced a new generalized alpha power inverse Weibull distribution. The author has used alpha power transformation method in order to yield this distribution. Various characterization results and statistical properties of the alpha power inverse Weibull distribution have been proposed. Finally, the author employed several estimation techniques for estimation of the unknown parameters of alpha power inverse Weibull distribution. 
\cite{hassan2019inverse} proposed a new three-parameter lifetime distribution named as the inverse power Lomax distribution. This distribution can be obtained using inverse transformation from power Lomax distribution. 
\cite{rao2019exponentiated} introduced a generalization of the inverse Rayleigh distribution known as exponentiated inverse Rayleigh distribution. Some statistical properties of the newly proposed distribution have been investigated. In particular, the mode, quantiles, moments, reliability, and hazard functions have been obtained. Various estimates for the model parameters have been proposed.

In statistical literature, inverted distributions have been generated by using the inverse transformation to display different properties of the probability density function and hazard rate functions and also allowed to explore more applicability in real-life phenomenon. The inverted distributions are applied in several areas related to econometrics, biological sciences, medical research, and survival analysis. To know more details about inverted distributions, one may be referred to \cite{sheikh1987some} and \cite{lehmann2012inverted}. There are some well-known inverted distributions and these are listed below:
\begin{itemize}
	\item Inverted Weibull distribution with cumulative distribution function (CDF): $F_{IW}(x)=\exp(-\alpha x^{-\theta}),~x,\alpha,\theta>0$. (For instance, see \cite{keller1982alternate})
	\item Inverted Rayleigh distribution with CDF: $F_{IR}(x)=\exp(-\alpha x^{-2}),~x,\alpha>0$. (For instance, see \cite{voda1972inverse}) 
	\item Inverted Gamma distribution with CDF: $F_{IG}(x)= \Gamma(\alpha,\theta x^{-1})\Gamma(\alpha),~x,\alpha,\theta>0$. (For instance, see \cite{lin1989inverted})
	\item Inverted power Lindley distribution with CDF: $F_{IPL}(x)=(1+\frac{\lambda x^{-\alpha}}{1+\lambda}) exp(-\lambda x^{-\alpha}),~x,\alpha,\lambda>0$. (For instance, see \cite{barco2017inverse}) 
\end{itemize} The purpose of this paper is two folds. First, we introduce a new distribution and study its various statistical properties. Secondly, estimators for the unknown parameters of IGLFR distribution have been obtained from both frequentist and Bayesian approaches and  establish a recommendation for choosing the most appropriate estimator, which we believe would be of significant interest to applied statisticians. The maximum likelihood estimates (MLEs) and maximum product spacing estimates (MPSEs) have been considered as frequentist estimates. Furthermore, the Bayes estimates of the unknown parameters have been computed under the assumptions of independent gamma priors based on squared error loss function (SELF). A Monte Carlo simulation study has been carried out to compare the performance of the estimates. Finally, two real life data sets have been analyzed for illustrative purpose.

The article is organized as follows. In Section $2$, we first provide construction of the newly proposed IGLFR distribution. Some special cases are presented. Then, several statistical properties are explored. In Section $3$, we present estimation techniques for the computation of the unknown parameters. In Section $4$,  simulation study has been carried out to see  comparative performance of the proposed estimates. Section $5$ presents real life data analysis. Finally, some concluding remarks are added in Section $6$.

\section{The proposed IGLFR distribution}
In this section, we introduce a new continuous distribution referred to as the IGLFR distribution and study its properties.
The new three-parameter distribution is derived using an inverse transformation to the generalized linear failure rate (GLFR) distribution proposed by \cite{sarhan2009generalized}. The cumulative distribution function (CDF) of a GLFR distributed random variable $Y$ is given by 
\begin{eqnarray}\label{eq2.1}
G(y;\alpha,\beta,\theta)=\left[1-\exp\left\{-\left(\alpha y+\frac{\beta y^2}{2}\right)\right\}\right]^{\theta},~y>0,~\alpha,\beta,\theta>0.
\end{eqnarray}
Note that in (\ref{eq2.1}), $\theta$ is known as the shape parameter. The probability density function (PDF) of  GLFR distribution is
\begin{eqnarray}
g(y;\alpha,\beta,\theta)=\theta(\alpha+\beta y)\left[1-\exp\left\{-\left(\alpha y+\frac{\beta y^2}{2}\right)\right\}\right]^{\theta-1} \exp\left\{-\left(\alpha y+\frac{\beta y^2}{2}\right)\right\},~y>0.
\end{eqnarray}

\subsection{Definition}
Let $X=\frac{1}{Y}$ be a transformation proposed by \cite{tahir2018inverted} and here $Y$ follows GLFR distribution with CDF (\ref{eq2.1}). Then, the CDF of $X$ is 
\begin{eqnarray}\label{eq2.3}
F(x)&=&P(X\le x)=P\left(\frac{1}{Y}\le x\right)=P\left(Y\ge \frac{1}{x}\right)=1-P\left(Y< \frac{1}{x}\right)\nonumber\\&=&1-\left[1-\exp\left\{-\left(\frac{\alpha}{x}+\frac{\beta}{2x^2}\right)\right\}\right]^{\theta},~~x>0.
\end{eqnarray}
Thus, the PDF of $X$ can be easily obtained from (\ref{eq2.3}), and is given by
\begin{eqnarray}\label{eq2.4}
f(x;\alpha,\beta,\theta)=\theta\left(\frac{\alpha}{x^2}+\frac{\beta}{x^3}\right)\exp\left\{-\left(\frac{\alpha}{x}+\frac{\beta}{2x^2}\right)\right\}\left[1-\exp\left\{-\left(\frac{\alpha}{x}
+\frac{\beta}{2x^2}\right)\right\}\right]^{\theta-1},
\end{eqnarray}
where $x>0$ and $\alpha,~\beta,~\theta>0.$ We call the model associated with CDF and PDF given by (\ref{eq2.3}) and (\ref{eq2.4}), respectively as the IGLFR distribution, and denote $X\sim IGLFR(\alpha,\beta,\theta)$. The PDF and CDF of $IGLFR(\alpha,\beta,\theta)$ distribution are plotted in Figure $1$  and Figure $2$, respectively for various parameter values. 
\begin{figure}[h]
	\centering
	\includegraphics[width=5.5cm,height=5cm]{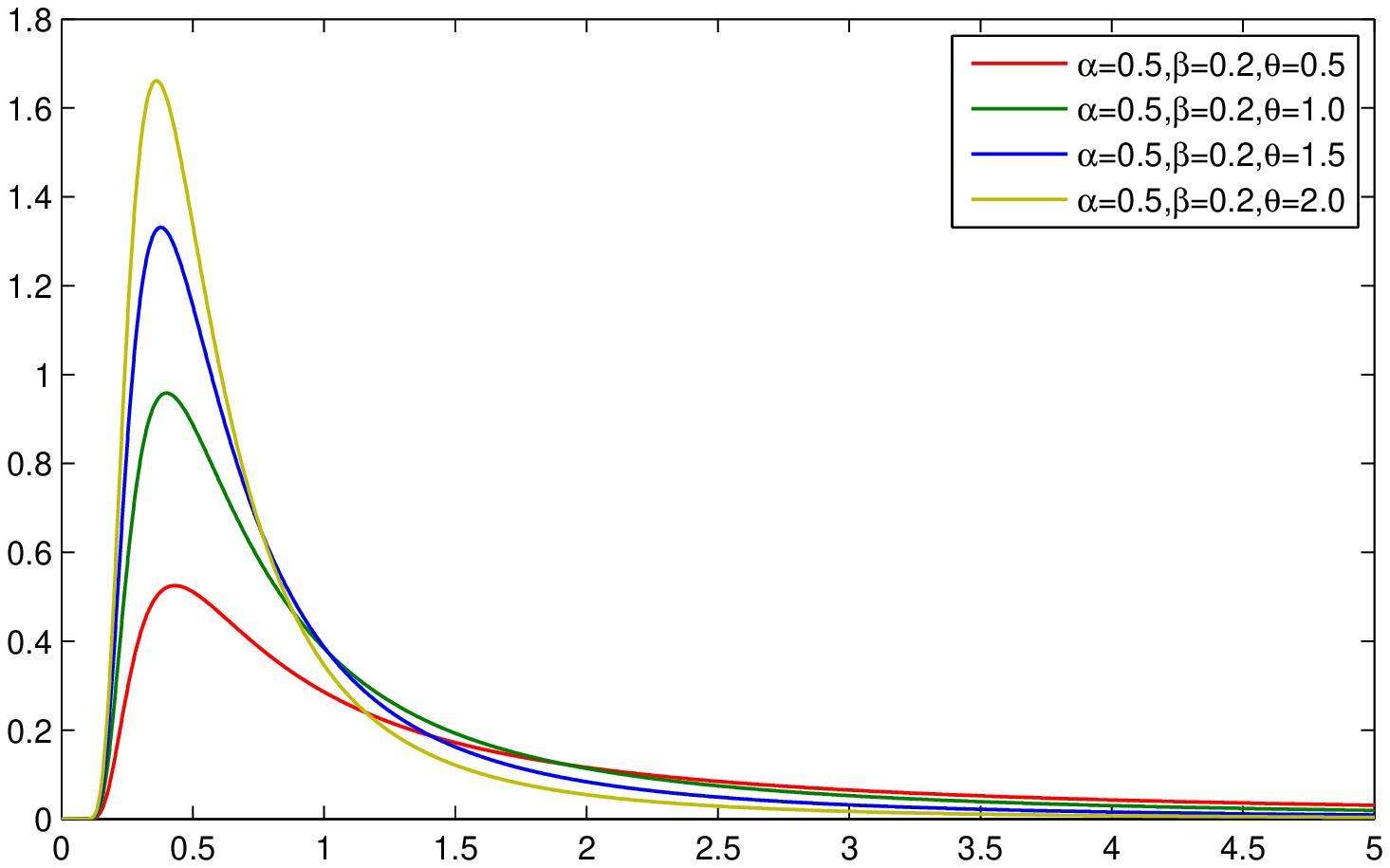}
	\includegraphics[width=5.5cm,height=5cm]{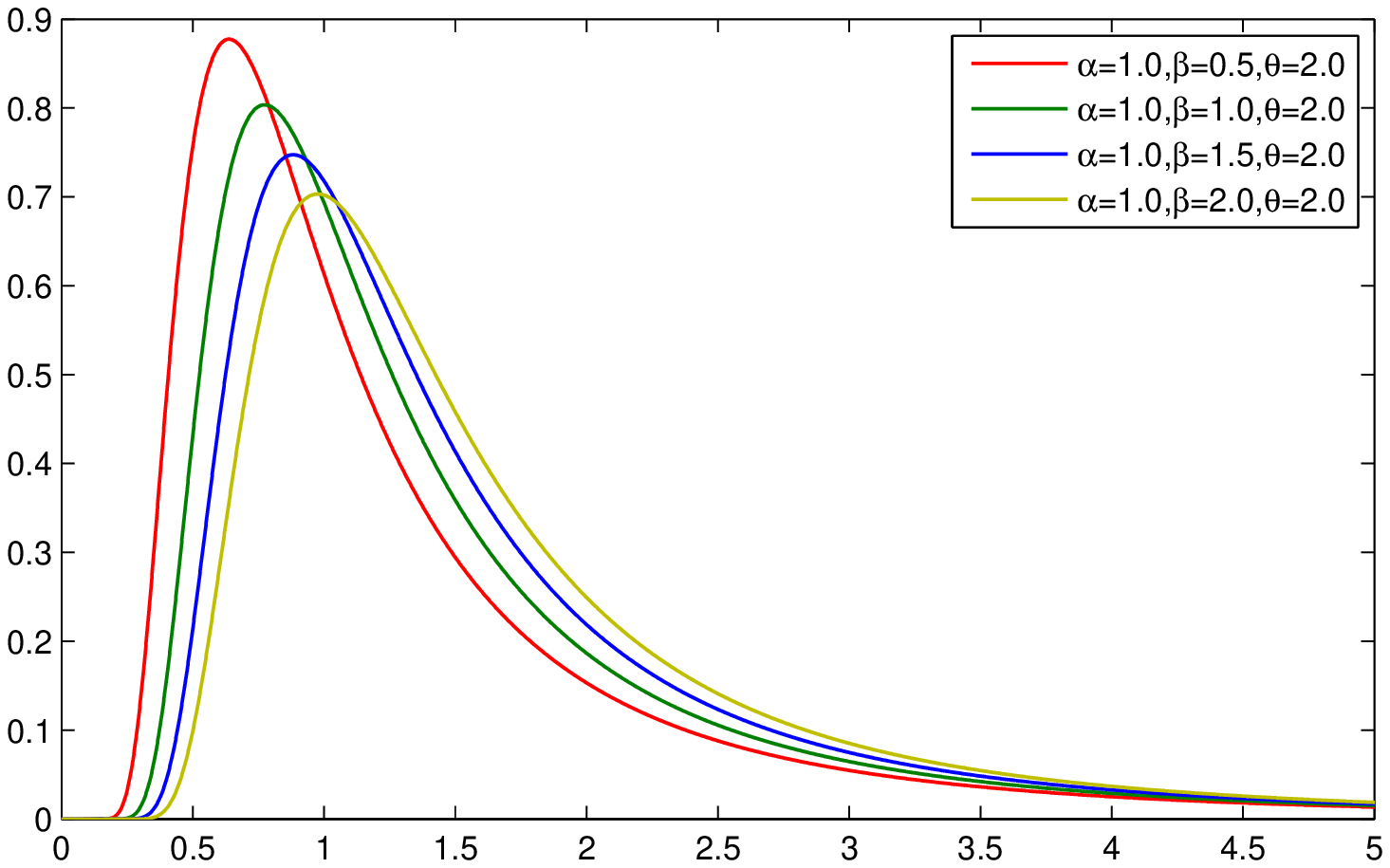}
	\includegraphics[width=5.5cm,height=5cm]{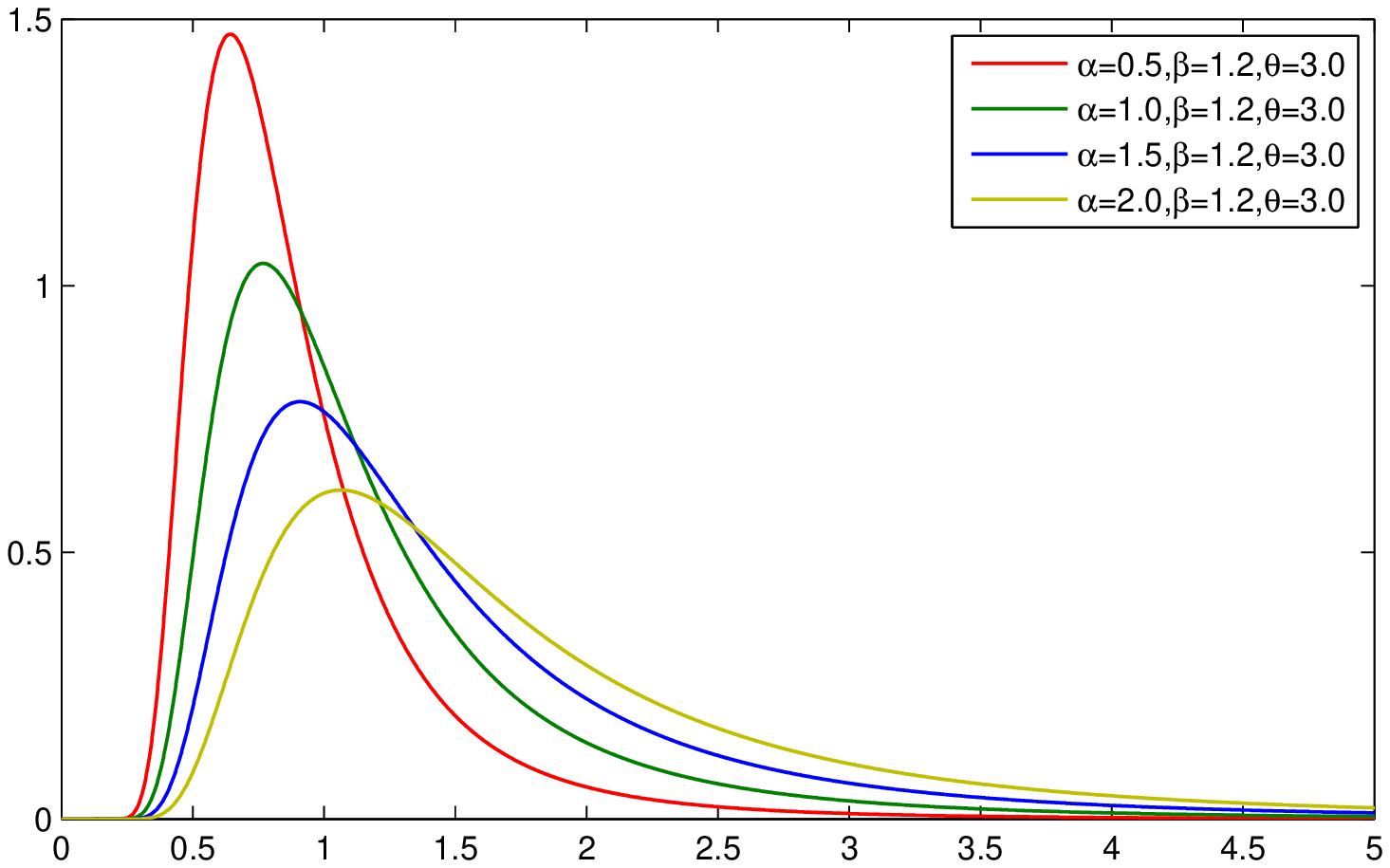}
	\caption{PDFs of IGLFR($\alpha$,$\beta$,$\theta$) distribution, for different values of parameters.}
	\label{fig:PdfofIGLFRD}
\end{figure}

\begin{figure}[h]
	\centering
	\includegraphics[width=5.5cm,height=5cm]{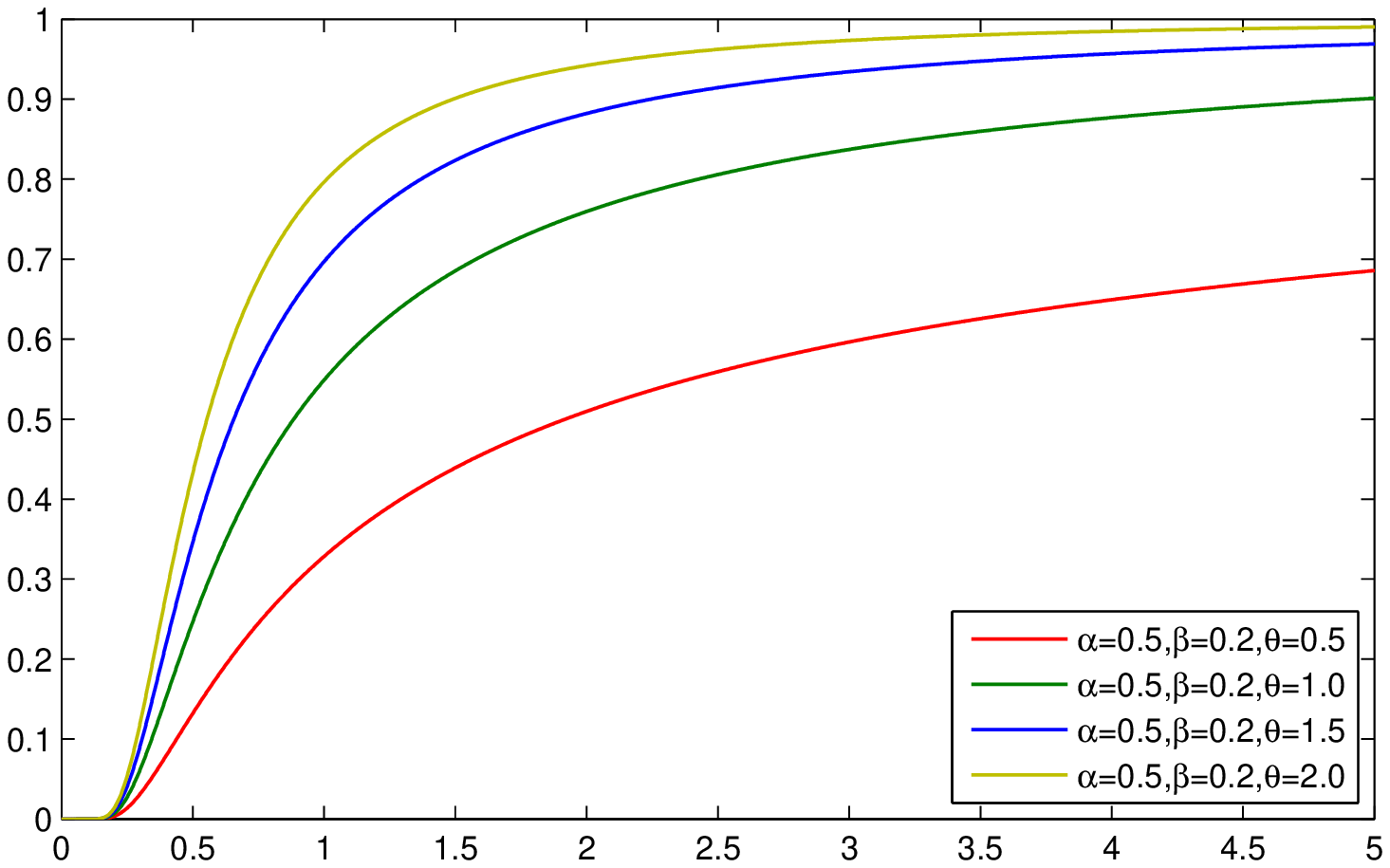}
	\includegraphics[width=5.5cm,height=5cm]{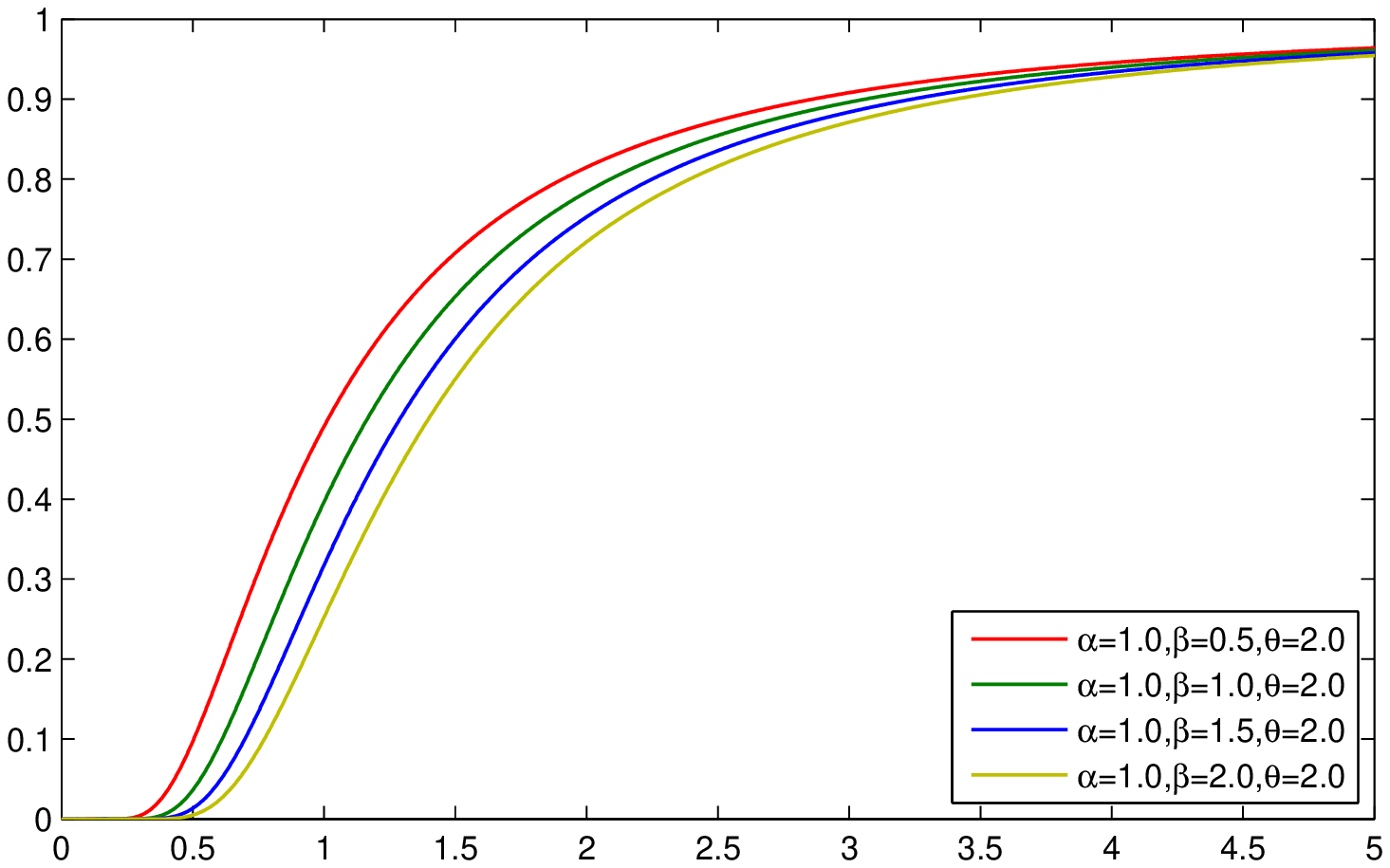}
	\includegraphics[width=5.5cm,height=5cm]{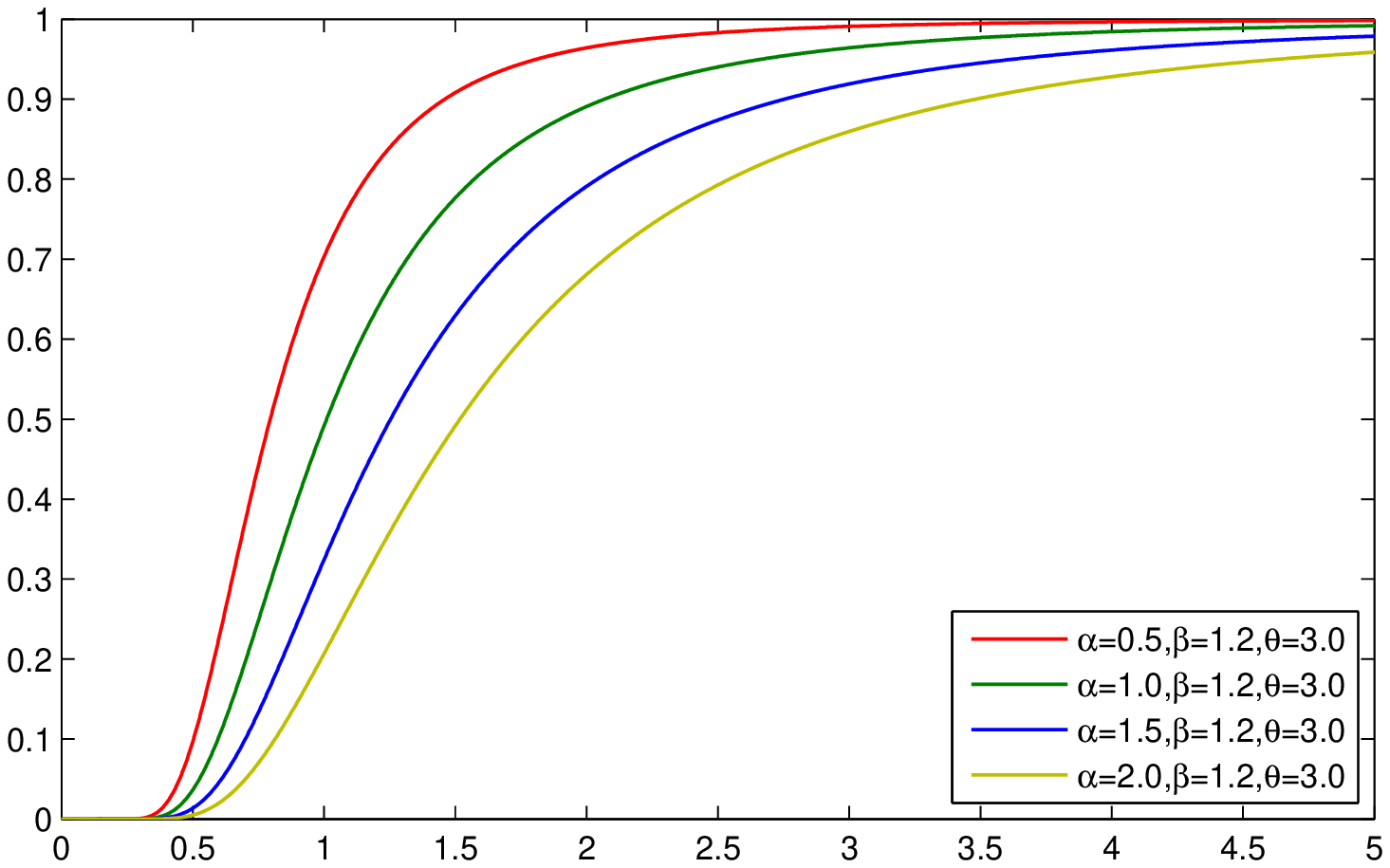}
	\caption{CDFs of IGLFR($\alpha$,$\beta$,$\theta$) distribution, for different values of parameters.}
	\label{fig:CdfofIGLFRD}
\end{figure}

{\bf Some special cases:}
\begin{itemize}
	\item Let $\alpha=0$, $\beta=\frac{2}{\lambda^2},$ and $\theta=\alpha>0.$ Then, the IGLFR distribution with PDF given by (\ref{eq2.4}) reduces to the PDF of generalized inverted Rayleigh distribution (see \cite{bakoban2015estimation}) .
	
	\item Consider $\alpha=\lambda>0$, $\beta=0$, and $\theta=\alpha>0.$ Using these particular values of the parameters, the PDF of IGLFR distribution becomes the PDF of generalized inverted exponential distribution (see \cite{krishna2013reliability}).
	
	\item Let $\alpha=0,$ $\beta=2\sigma^2$, and $\theta=\alpha>0$. Then, the IGLFR distribution reduces to the exponentiated inverse Rayleigh distribution (see \cite{rao2019exponentiated}).
	
	\item For $\alpha=0,$ $\beta=\lambda$ and $\theta=1$, the IGLFR distribution reduces to inverse Rayleigh distribution (see \cite{malik2018new}).
	
	\item Let $\alpha=\alpha>0,$ $\beta=\theta>0$, and $\theta=1$. Then, the IGLFR distribution becomes  modified inverse Rayleigh distribution (see \cite{khan2014modified}).  
	
	\item Consider $\alpha=\frac{1}{\beta}$, $\beta=0$, and $\theta=1.$ Then, the IGLFR distribution reduces to inverted exponential distribution (see \cite{dey2007inverted}).
\end{itemize}

\begin{remark}
For a positive integer $\theta$,  the CDF of $IGLFR(\alpha,\beta,\theta)$ distribution represents the CDF of the minimum of a simple random sample of size $\theta$ from the modified inverse Rayleigh distribution (MIRD) with survival function given by (see \cite{khan2014modified})
\begin{eqnarray*}
\bar{K}(x)=1-\exp\left\{-\left(\frac{\alpha}{x}+\frac{\beta}{2x^2}\right)\right\},~~x>0,~\alpha,\beta>0.
\end{eqnarray*}
Thus, we can say that the $IGLFR(\alpha,\beta,\theta)$ distribution provides the CDF of a series system when each component of the system has the MIRD.
\end{remark}

\begin{remark}
	Let $T=X^{\frac{1}{\lambda}}$, where $\lambda>0$ and $X\sim IGLFR(\alpha,\beta,\theta)$. Then, the CDF of $T$ is obtained as $$K_{T}(x)=1-\left[1-\exp\left\{-\left(\frac{\alpha}{x^{\lambda}}+\frac{\beta}{2x^{2\lambda}}\right)\right\}\right]^\theta,~~x>0.$$
	To the best of our knowledge, the CDF of $T$ does not belong to a known family of distributions. However, if $\beta=0,$ then the CDF of $T$ reduces to CDF of the well-known inverse exponentiated Weibull distribution. 
\end{remark}

\subsection{Hazard rate, reversed hazard rate and odd functions}
In survival analysis, the hazard function plays a central role. Suppose at time $t$, an item is working. Under this prior information, the hazard function represents  the failure probability of the item in next $dt$ units of time.  The hazard rate function uniquely determines a CDF. The hazard function of IGLFR distribution is obtained as
\begin{equation}\label{eq2.5*}
\begin{aligned}
h(x;\alpha,\beta,\theta)=\frac{\theta ~ {(\frac{\alpha}{x^2}+\frac{\beta}{x^3})} ~ {e^{-(\frac{\alpha}{x}+\frac{\beta}{2x^2})}}}{1-e^{-(\frac{\alpha}{x}+\frac{\beta}{2x^2})}}=\theta~ h(x;\alpha,\beta,1).
\end{aligned}
\end{equation}
In the context of reliability or maintenance applications, the hazard rate function describes how likely is the occurence of a failure in the immediate future. From (\ref{eq2.5*}), it is clear that the IGLFR distribution is a proportional hazard family. The plots of $h(x;\alpha,\beta,\theta)$ for different choices of the parameters are depicted in Figure $3$.
\begin{figure}[h]
	\centering
	\includegraphics[width=5.5cm,height=5cm]{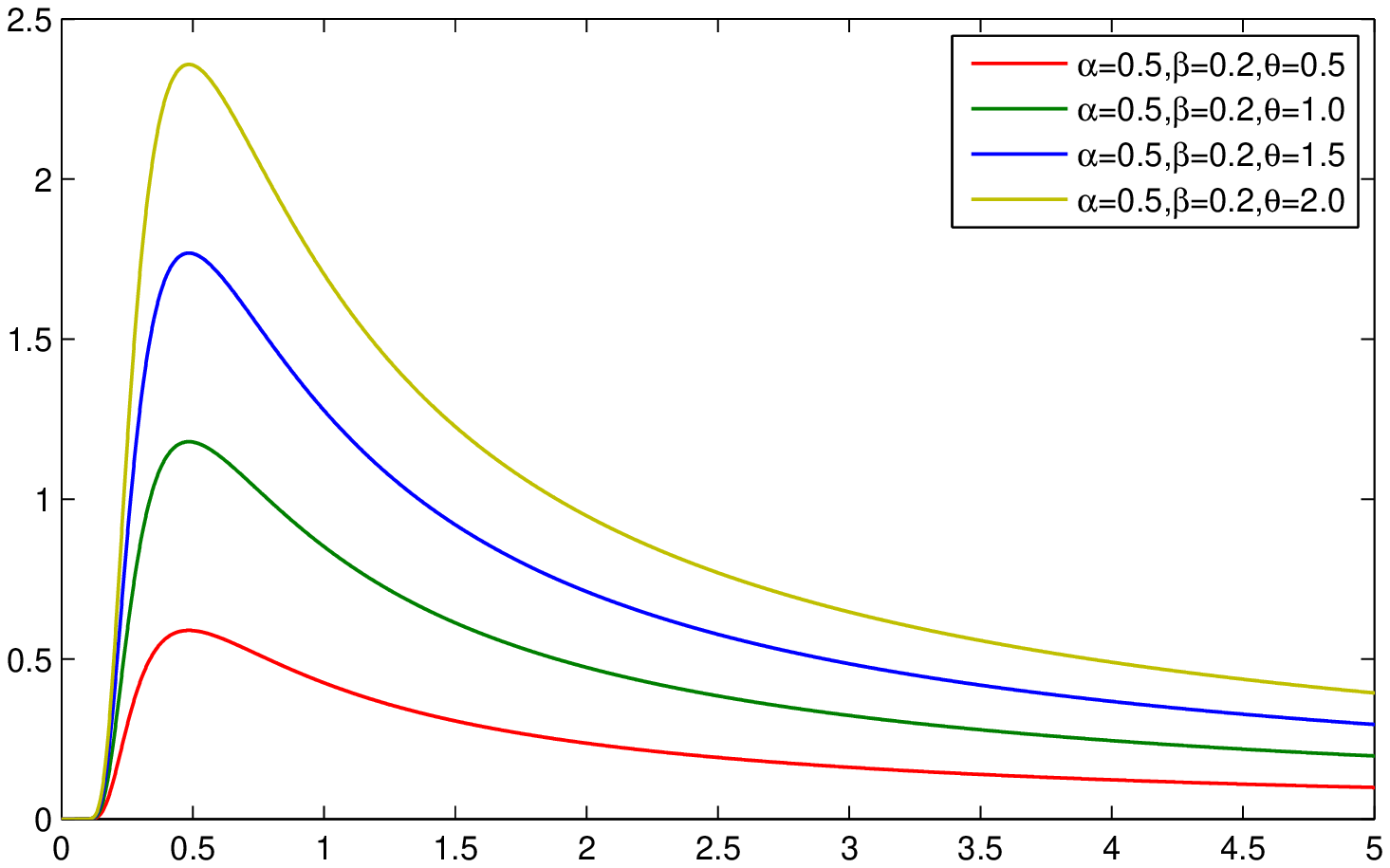}
	\includegraphics[width=5.5cm,height=5cm]{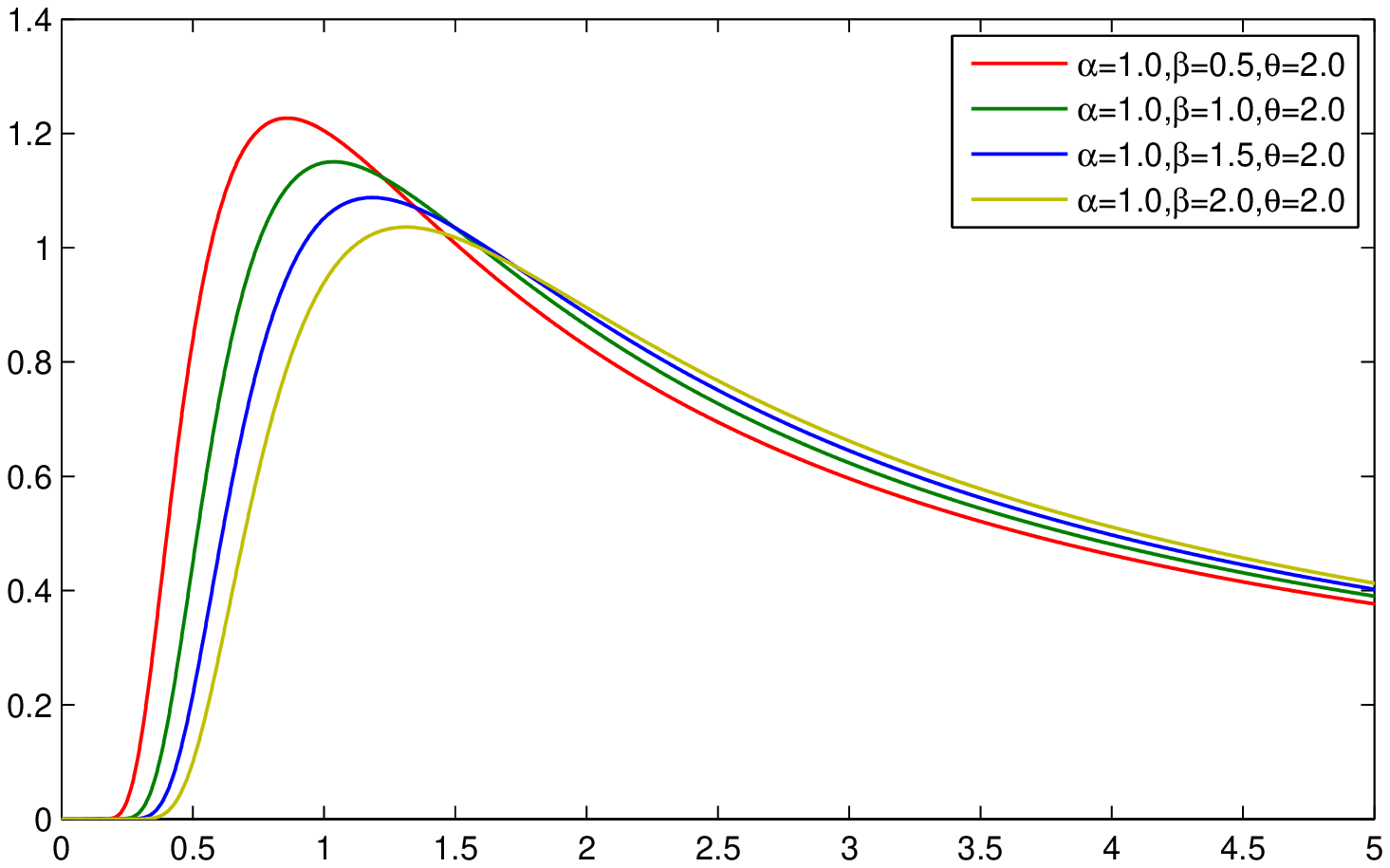}
	\includegraphics[width=5.5cm,height=5cm]{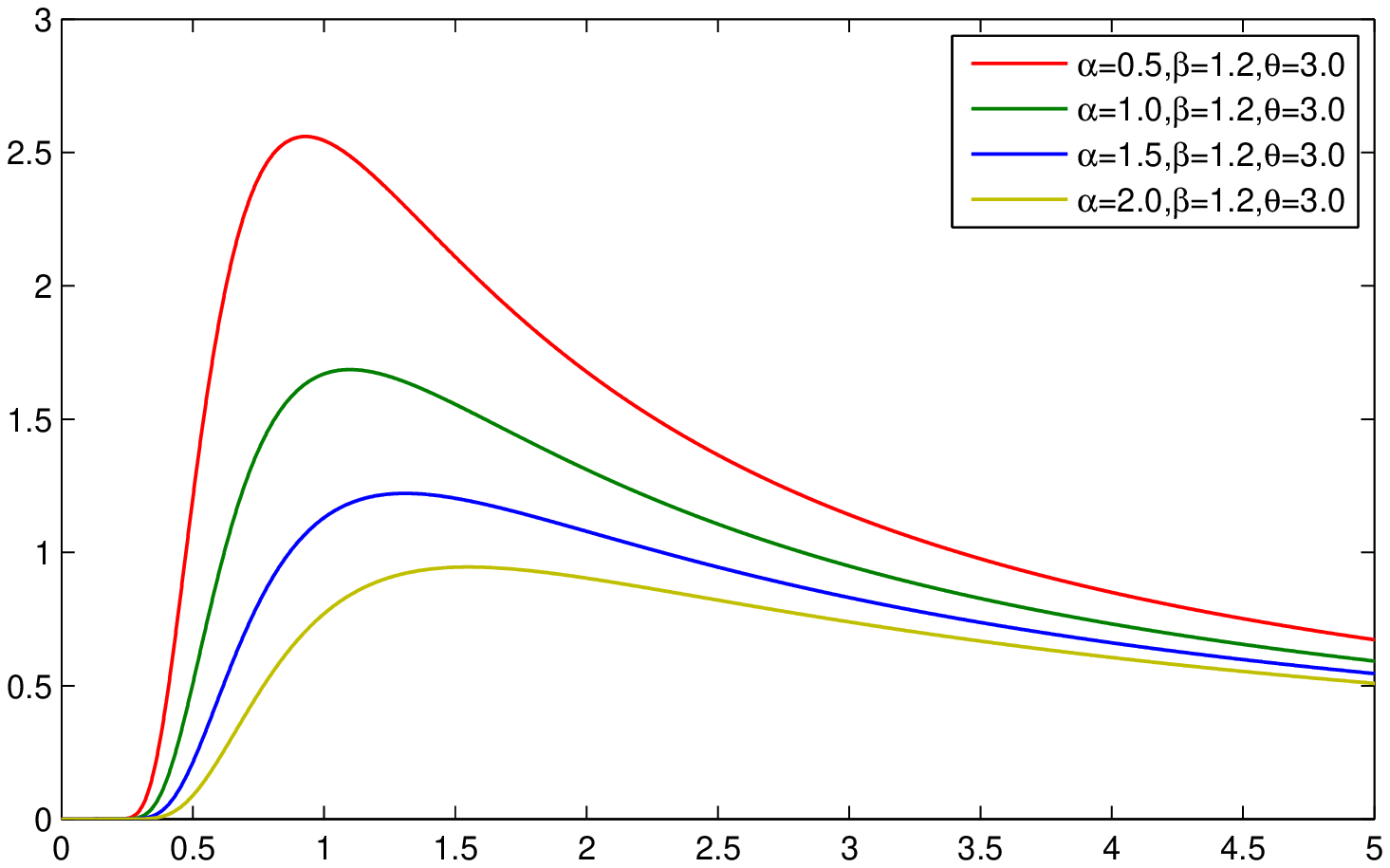}
	\caption{Hazard rate functions of IGLFR($\alpha$,$\beta$,$\theta$), for different values of the parameters.}
	\label{fig:CdfofIGLFRD*}
\end{figure}
As a dual of  hazard function, the reversed hazard rate function characterizes the probability of an immediate past failure, under the information that a failure has already occured. The reversed hazard rate function of a random variable $X$ following IGLFR($\alpha$,$\beta$,$\theta$) is given by 
\begin{equation}
\begin{aligned}
\text{r}(x;\alpha,\beta,\theta)=  \frac{\theta ~{(\frac{\alpha}{x^2}+\frac{\beta}{x^3})} ~ {e^{-(\frac{\alpha}{x}+\frac{\beta}{2x^2})}} ~ \Big[1-e^{-(\frac{\alpha}{x}+\frac{\beta}{2x^2})}\Big]^{\theta-1}}{1-\Big(1-e^{-(\frac{\alpha}{x}+\frac{\beta}{2x^2})}\Big)^ \theta}, ~ x>0.
\end{aligned}
\end{equation}
The reversed hazard rate function of the IGLFR($\alpha$,$\beta$,$\theta$) distribution is plotted in Figure $4$, for different choices of the parameters. 
\begin{figure}[h]
	\centering
	\includegraphics[width=5.5cm,height=5cm]{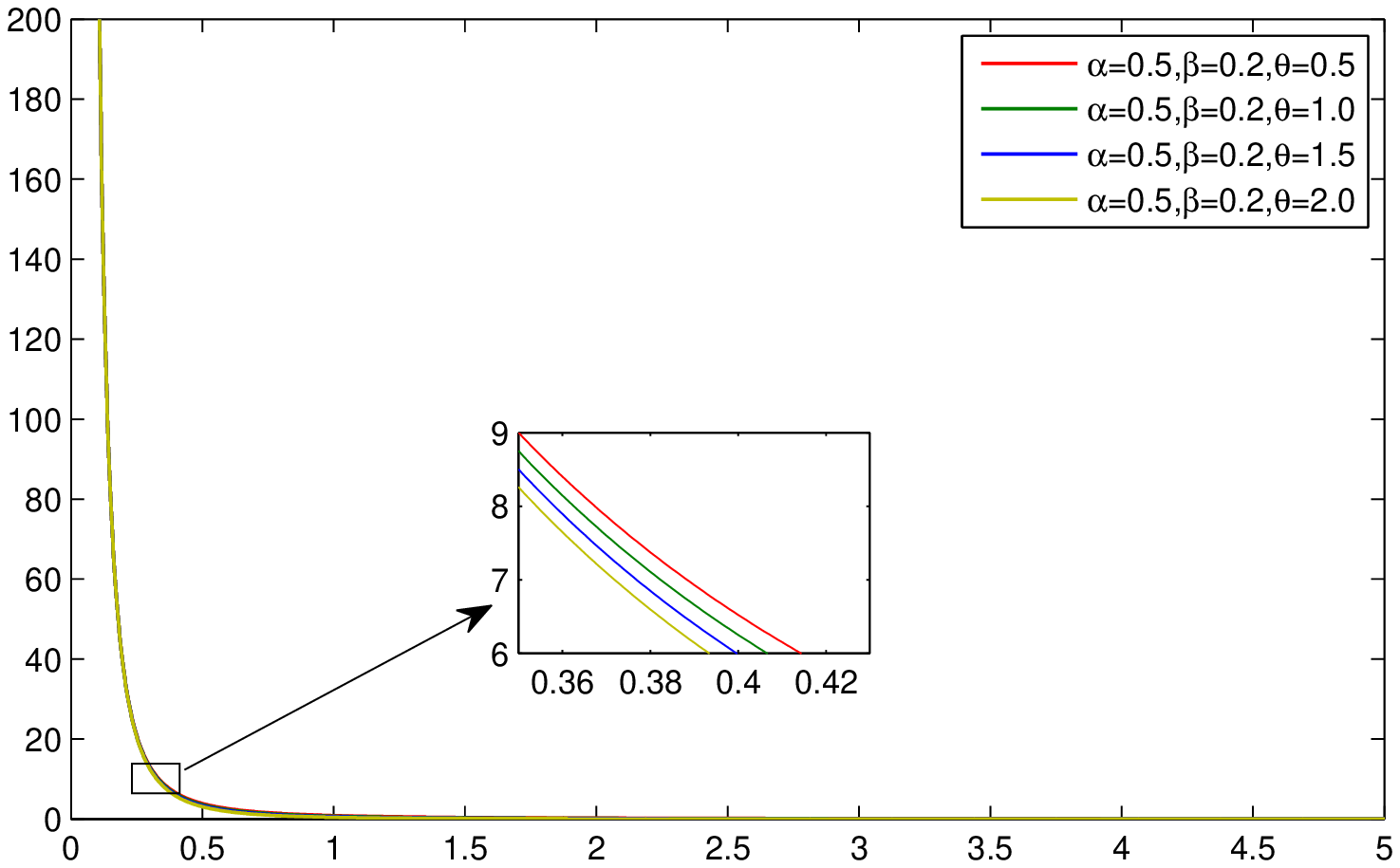}
	\includegraphics[width=5.5cm,height=5cm]{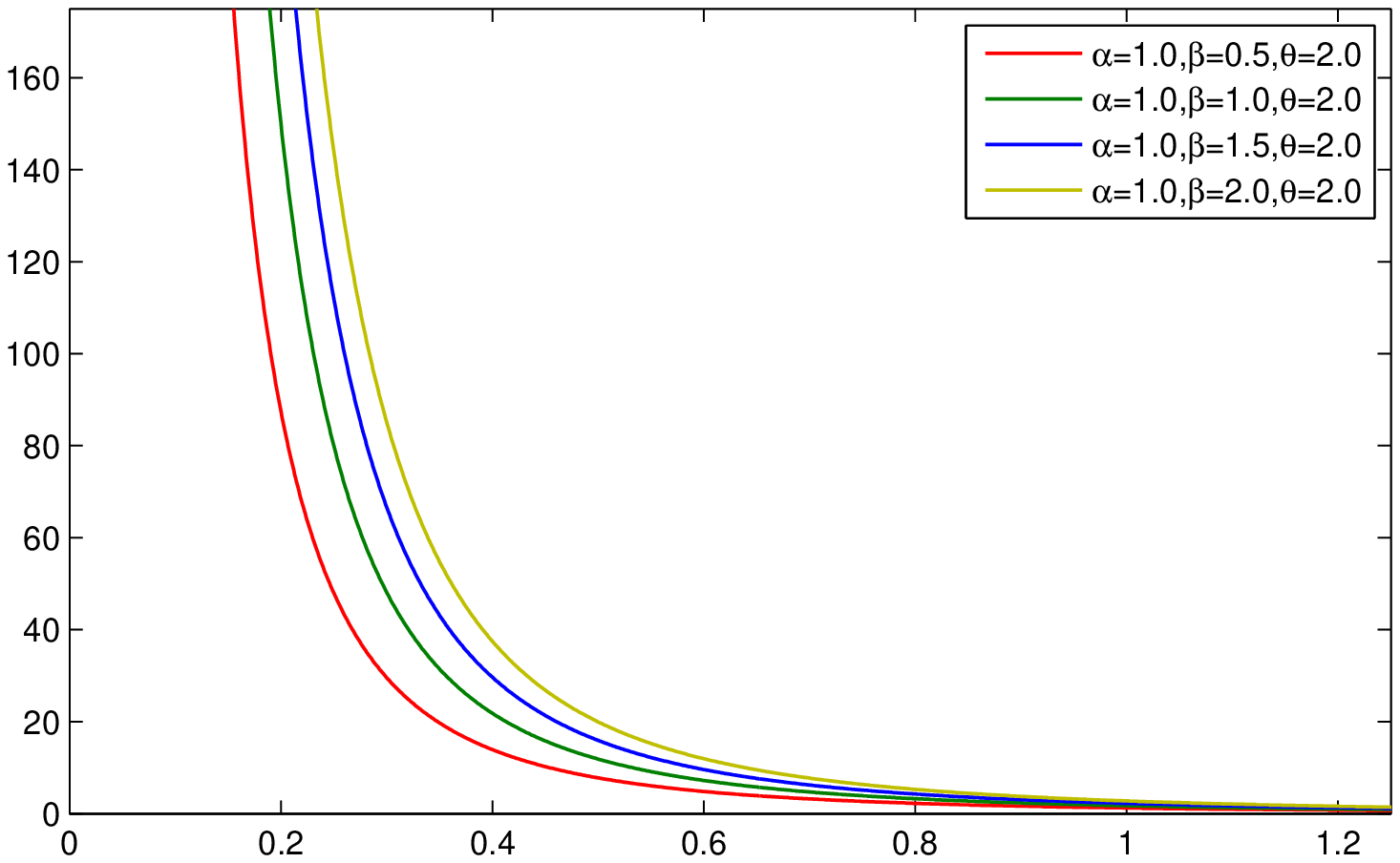}
	\includegraphics[width=5.5cm,height=5cm]{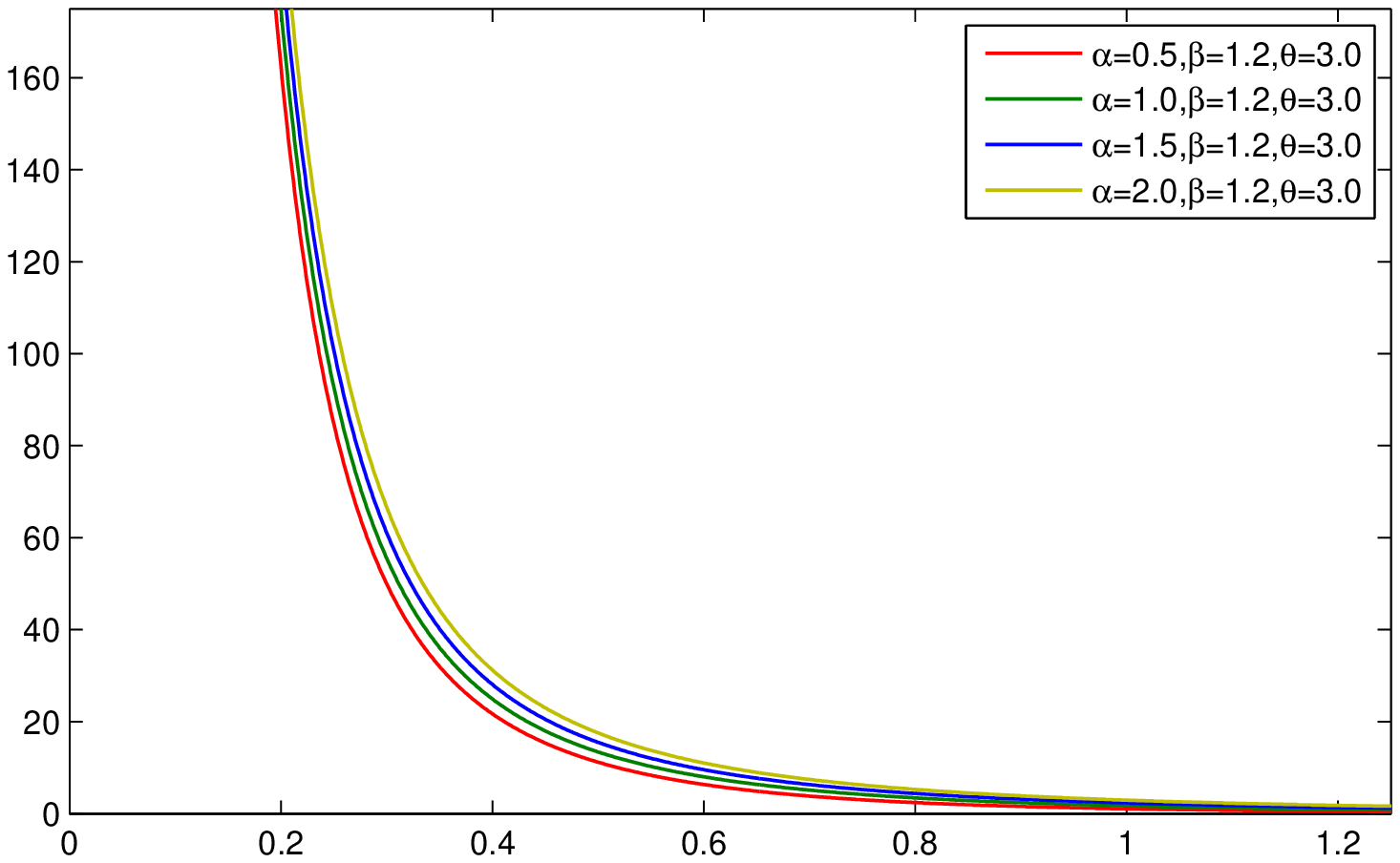}
	\caption{Reversed hazard rate functions of IGLFR($\alpha$,$\beta$,$\theta$).}
	\label{fig:CdfofIGLFRD**}
\end{figure}
The odd function is 
\begin{equation}
\begin{aligned}
\text{O}(x;\alpha,\beta,\theta)=\frac{1}{\Big[1-e^{-(\frac{\alpha}{x}+\frac{\beta}{2x^2})}\Big]^ \theta}-1, ~ x>0.
\end{aligned}
\end{equation}

\subsection{Asymptotic behaviour of the PDF and hazard function}
Here, we will study behaviour of the PDF and hazard  function of the newly proposed model when $x$ tends to $0$ and $\infty.$ After some calculations, we obtain
\begin{eqnarray}
\lim_{x\rightarrow 0} f(x)&=&\lim_{x\rightarrow 0}\left[ \theta\left(\frac{\alpha}{x^2}+\frac{\beta}{x^3}\right)\exp\left\{-\left(\frac{\alpha}{x}+\frac{\beta}{2x^2}\right)\right\}\left[1-\exp\left\{-\left(\frac{\alpha}{x}
+\frac{\beta}{2x^2}\right)\right\}\right]^{\theta-1}\right]\nonumber\\
&=&0
\end{eqnarray}
and 
\begin{eqnarray}
\lim_{x\rightarrow \infty} f(x)&=&\lim_{x\rightarrow 0}\left[ \theta\left(\frac{\alpha}{x^2}+\frac{\beta}{x^3}\right)\exp\left\{-\left(\frac{\alpha}{x}+\frac{\beta}{2x^2}\right)\right\}\left[1-\exp\left\{-\left(\frac{\alpha}{x}
+\frac{\beta}{2x^2}\right)\right\}\right]^{\theta-1}\right]\nonumber\\
&=&0.
\end{eqnarray}
These results show that the proposed model is non-monotonic and has atleast one mode. To verify this, we use graphical approach (see Figure $1$), which indicates that the proposed model has only one mode, that is, the IGLFR distribution is unimodal.  Further,
\begin{eqnarray}
\lim_{x\rightarrow 0} \text{h}(x;\alpha,\beta,\theta)&=&\lim_{x\rightarrow 0}\left[\frac{\theta ~ {(\frac{\alpha}{x^2}+\frac{\beta}{x^3})} ~ {e^{-(\frac{\alpha}{x}+\frac{\beta}{2x^2})}}}{1-e^{-(\frac{\alpha}{x}+\frac{\beta}{2x^2})}}\right]\nonumber\\
&=&0
\end{eqnarray}
and 
\begin{eqnarray}
\lim_{x\rightarrow \infty} \text{h}(x;\alpha,\beta,\theta)&=&\lim_{x\rightarrow \infty}\left[\frac{\theta ~ {(\frac{\alpha}{x^2}+\frac{\beta}{x^3})} ~ {e^{-(\frac{\alpha}{x}+\frac{\beta}{2x^2})}}}{1-e^{-(\frac{\alpha}{x}+\frac{\beta}{2x^2})}}\right]\nonumber\\
&=&0.
\end{eqnarray}
From the above limits, similar observation for the hazard function can be noticed, that is, it has upside down bathtub shape. It is worth mentioning that the GLFR model has increasing, decreasing and bathtub shaped hazard functions (see \cite{sarhan2009generalized}). However, the newly proposed model has upside down bathtub shape hazard function. There are various real life examples, where distributions with upside down bathtub shape hazard function is useful for  analyse. For example, in biology, the upside down bathtub-shaped hazard rate model can be used in the course of a disease whose mortality reaches a peak after some time, and then decreases gradually.

\subsection{Moments, median, quantile and mode}
In this subsection, we will discuss about moments, median, quantile and mode of the newly proposed IGLFR distribution. As expected, the mean and other moments are difficult to derive in closed form. The $r$th order moment of the IGLFR distribution is given by
\begin{eqnarray}\label{eq2.5}
\mu_{r}'&=&\int_{0}^{\infty} x^r f(x) dx\nonumber\\
&=& \int_{0}^{\infty} x^r  \theta\left(\frac{\alpha}{x^2}+\frac{\beta}{x^3}\right)\exp\left\{-\left(\frac{\alpha}{x}+\frac{\beta}{2x^2}\right)\right\}\left[1-\exp\left\{-\left(\frac{\alpha}{x}
+\frac{\beta}{2x^2}\right)\right\}\right]^{\theta-1}dx\nonumber\\
&=& \int_{0}^{\infty} x^{r-3}  \theta\left(\alpha x+\beta\right)\exp\left\{-\left(\frac{\alpha}{x}+\frac{\beta}{2x^2}\right)\right\}\left[1-\exp\left\{-\left(\frac{\alpha}{x}
+\frac{\beta}{2x^2}\right)\right\}\right]^{\theta-1}dx.
\end{eqnarray}
Further, we know that 
$$\left[1-\exp\left\{-\left(\frac{\alpha}{x}
+\frac{\beta}{2x^2}\right)\right\}\right]^{\theta-1}=\sum_{i=0}^{\infty}(-1)^{i} {\theta-1 \choose i}\exp\left\{-i\left(\frac{\alpha}{x}+\frac{\beta}{2x^2}\right)\right\}.$$
Substituting this in (\ref{eq2.5}), we obtain 
\begin{eqnarray}\label{eq2.6}
\mu_{r}'&=& \int_{0}^{\infty} x^{r-3}  \theta\left(\alpha x+\beta\right)\sum_{i=0}^{\infty}(-1)^{i} {\theta-1 \choose i}\exp\left\{-(i+1)\left(\frac{\alpha}{x}+\frac{\beta}{2x^2}\right)\right\}dx\nonumber\\
&=& \theta \sum_{i=0}^{\infty}(-1)^{i} {\theta-1 \choose i} \int_{0}^{\infty} x^{r-3}  \left(\alpha x+\beta\right)\exp\left\{-(i+1)\left(\frac{\alpha}{x}+\frac{\beta}{2x^2}\right)\right\}dx.
\end{eqnarray}

Let the median of $X$ be M. Then, M must satisfy the following equation:
\begin{equation}
\int_{0}^{\text{M}} f(x;\alpha,\beta,\theta)~dx  = \int_{\text{M}}^{\infty} f(x;\alpha,\beta,\theta) ~dx.
\end{equation}
Substituting the PDF given by (\ref{eq2.4}) in the last equation, and then after some calculations, we obtain  
\begin{equation}
\text{Median } = \frac{-\alpha - \sqrt{\alpha^2-2\beta \log\bigg[1 - \Big(\frac{1}{2}\Big)^\frac{1}{\theta}\bigg]} }{2\log\bigg[1 - \Big(\frac{1}{2}\Big)^\frac{1}{\theta}\bigg]}.
\end{equation}\\
Suppose U $\thicksim$ Uniform (0,1), then to generate quantile from IGLFR distribution, we invert the CDF of IGLFR(x;$\alpha,\beta,\theta$).
Let  $x=F^{-1}(q)$ be quantile of the IGLFR(x;$\alpha,\beta,\theta$) then 
from the  equation $F(x;\alpha,\beta,\theta) = q$, the $q$th quantile is turned out as \\
\begin{equation}
	\ \ x_{q}= \frac{-\alpha - \sqrt{\alpha^2-2\beta \log\bigg[1 -\big(1-q\big)^\frac{1}{\theta}\bigg]} }{2\log\bigg[1 -\big(1-q\big)^\frac{1}{\theta}\bigg]},\  \text{since}\ 0<x<1.
\end{equation}
	There are different methods which are used to find skewness and kurtosis in a certain distribution. Quantile based measures, such as Bowley skewness
[\cite{kenney1939mathematics}] and Moors kurtosis [\cite{moors1988quantile}], can quantify asymmetry and the peakedness of a given distribution. Based on the quantiles, the coefficients of skewness and kurtosis are provided by\\
\begin{equation}\nonumber
	\begin{aligned}
		\text{Skewness} =&\frac{(Q_3-Q_2)-(Q_2-Q_1)}{(Q_3-Q_2)+(Q_2-Q_1)}\\
		\text{Kurtosis} =&\frac{(E_7-E_5)-(E_3-E_1)}{(E_6-E_2)}.
	\end{aligned}
\end{equation}
where $Q_{i}$ is the $i^{th}$ quartile, $Q_i=F^{-1}(\frac{i}{4})$ and $E_{i}$ is the $i^{th}$ octile, $E_i=F^{-1}(\frac{i}{8})$.

The mode of the distribution can be obtained using the following approach: 
\begin{equation}\nonumber
\begin{aligned}
f'(x)&=0,\\
f''(x) &< 0.
\end{aligned}
\end{equation}
The first order derivative of the PDF $f(x)$ given by (\ref{eq2.4}) with respect to $x$ is given by 
\begin{equation}
\begin{aligned}
\theta ~{e^{-(\frac{\alpha}{x}+\frac{\beta}{2x^2})}} ~ \Big[1-e^{-(\frac{\alpha}{x}+\frac{\beta}{2x^2})}\Big]^{\theta-2} ~ \Bigg[\bigg(\frac{-2\alpha x^3 - 3\beta x^2 + \alpha^2 x^2 +\beta^2 + 2\alpha\beta x}{x^6} \bigg)&\\ +e^{-(\frac{\alpha}{x}+\frac{\beta}{2x^2})}\bigg(\frac{2\alpha x^3 + 3\beta x^2 -\theta(\beta^2 +2\alpha \beta x + \alpha^2 x^2)}{x^6} \bigg) \Bigg]&=0.\label{eq2.10}
\end{aligned}    	
\end{equation} 
The root of (\ref{eq2.10}), which satisfies $f''(x) < 0$ is known as the mode of the IGLFR distribution. We note that it is difficult to get the roots of (\ref{eq2.10}) in explicit form. Thus, for some particular choices of the parameters, the mode can be obtained by adopting some numerical techniques. 

\subsection{Order statistics}
In this subsection, we derive explicit expression for the PDF of the $k$th order statistic $X_{k:n}$ in a random sample of size $n$ from the IGLFR distribution. Note that the $k$th order statistic has an important interpretation in reliability theory. The lifetime of a $(n-k+1)$-out-of-$n$ system is denoted by $X_{k:n}$. 
Using PDF and CDF of the IGLFR distribution, the CDF and PDF of the $k$th order statistic $X_{k:n}$ are respectively given by 
\begin{eqnarray}
F_{\text{k}:n}(x)&=& \sum_{i=k}^{n} {n \choose i}  \Bigg\{1-\Big[\gamma(\alpha,\beta,x)\Big]^ \theta\Bigg\}^{i} \Big[\gamma(\alpha,\beta,x)\Big]^{(n-i)\theta},\\
f_{\text{k}:n}(x)&=&\frac{n!}{(\text{k}-1)!(\text{n}-\text{k})!}~\theta~{(\frac{\alpha}{x^2}+\frac{\beta}{x^3})}~ {[1-\gamma(\alpha,\beta,x)]}~ \Big[\gamma(\alpha,\beta,x)\Big]^{\theta(n-\text{k}+1)-1}\nonumber\\
 &~&\times \Bigg\{1-\Big[\gamma(\alpha,\beta,x)\Big]^{\theta}\Bigg\}^{(\text{k}-1)}.
\end{eqnarray}
Denote $\gamma(\alpha,\beta,x)=1-\exp\{-(\frac{\alpha}{x}+\frac{\beta}{2x^2})\}.$
By substituting $k=1$ and $k=n$ in the last two equations, we respectively obtain the CDFs and PDFs of the minimum and maximum order statistics, which are given by

\begin{equation}\nonumber
\begin{aligned}
F_{1:n}(x)&= 1-\Big[\gamma(\alpha,\beta,x)\Big]^{n\theta},\\
f_{1:n}(x)&= n\theta~{\left(\frac{\alpha}{x^2}+\frac{\beta}{x^3}\right)}~ {[1-\gamma(\alpha,\beta,x)]}~ \Big[\gamma(\alpha,\beta,x)\Big]^{n\theta-1},
\end{aligned}
\end{equation}
and
\begin{equation}\nonumber
\begin{aligned}
F_{n:n}(x)&= \Bigg\{1-\Big[\gamma(\alpha,\beta,x)\Big]^ \theta\Bigg\}^{n}, \\
f_{n:n}(x)&= n\theta~{\left(\frac{\alpha}{x^2}+\frac{\beta}{x^3}\right)}~ {[1-\gamma(\alpha,\beta,x)]}~ \Big[\gamma(\alpha,\beta,x)\Big]^{(\theta-1)}~\Bigg\{1-\Big[\gamma(\alpha,\beta,x)\Big]^{\theta}\Bigg\}^{(n-1)},
\end{aligned}
\end{equation}
respectively. Note that the minimum and maximum order statistics respectively represent the lifetimes of a series and parallel systems.

\subsection{Stochastic ordering}
To study a comparative behaviour, the concept of stochastic ordering is a well-recognized tool in reliability theory. In this subsection, we discuss stochastic ordering results between two IGLFR distributed random variables. Let $X$ and $Y$ be two absolutely continuous random variables with PDFs $f_{X}(.)$, $f_{Y}(.),$ and CDFs $F_{X}(.)$, $F_{Y}(.),$ respectively. Then, $X$ is said to be smaller than $Y$ in the sense of 
\begin{itemize}
\item the likelihood ratio ordering,  abbreviated by $X\le_{lr}Y$, if $\frac{f_{Y}(x)}{f_{X}(x)}$ is non-decreasing in $x$;
\item the usual stochastic order, abbreviated by $X\le_{st}Y$, if $F_{Y}(x)\le F_{X}(x)$, for all $x;$
\item the hazard rate order, abbreviated by $X\le_{hr}Y$, if $h_{X}(x)\le h_{Y}(x),$ for all $x$, where $h_X$ and $h_Y$ are the hazard rate functions of $X$ and $Y,$ respectively;

\item the reverse hazard rate order, abbreviated by $X\le_{rh}Y$, if $r_{X}(x)\le r_{Y}(x),$ for all $x$, where $r_X$ and $r_Y$ are the reversed hazard rate functions of $X$ and $Y,$ respectively.
\end{itemize}
For various other stochastic orders and their applications, readers are referred to \cite{shaked2007stochastic}. The following result provides sufficient conditions such that the likelihood ratio order between two IGLFR distributed random variables exists. 

\begin{theorem}\label{th2.1}
	Let X $\thicksim$ IGLFR($\alpha_{1}$,$\beta_{1}$,$\theta_{1}$) and Y $\thicksim$ IGLFR($\alpha_{2}$,$\beta_{2}$,$\theta_{2}$). If $\alpha_{1}= \alpha_{2}=\alpha$, $\beta_{1}= \beta_{2}=\beta,$ and  $\theta_{1}\geq \theta_{2},$ then $X\leq_{lr}Y$.
\end{theorem}

\begin{proof}
We have
	\begin{eqnarray}\label{eq2.7}
	\frac{f_{Y}(x)}{f_{X}(x)}=\frac{\theta_{2}~{(\frac{\alpha_{2}}{x^2}+\frac{\beta_{2}}{x^3})}~ {e^{-(\frac{\alpha_{2}}{x}+\frac{\beta_{2}}{2x^2})}}~ \Big[1-e^{-(\frac{\alpha_{2}}{x}+\frac{\beta_{2}}{2x^2})}\Big]^{\theta_{2}-1}}{\theta_{1}~{(\frac{\alpha_{1}}{x^2}+\frac{\beta_{1}}{x^3})}~ {e^{-(\frac{\alpha_{1}}{x}+\frac{\beta_{1}}{2x^2})}}~ \Big[1-e^{-(\frac{\alpha_{1}}{x}+\frac{\beta_{1}}{2x^2})}\Big]^{\theta_{1}-1}}.
	\end{eqnarray}
	Further,
	\begin{eqnarray}\label{eq2.8}
	\frac{d}{dx}\text{log}\Bigg(\frac{f_{Y}(x)}{f_{X}(x)}\Bigg)&=& \frac{\alpha_{2}}{\alpha_{2}x+\beta_{2}}-\frac{\alpha_{1}}{\alpha_{1}x+\beta_{1}}+
	\frac{\alpha_{2}-\alpha_{1}}{x^{2}}-\frac{\beta_{2}-\beta_{1}}{x^{3}}\nonumber\\
	&~&-\frac{(\theta_{2}-1)~e^{-(\frac{\alpha_{2}}{x}+\frac{\beta_{2}}{2x^2})}~(\frac{\alpha_{2}}{x^2}+\frac{\beta_{2}}{x^3})}{1-e^{-(\frac{\alpha_{2}}{x}+\frac{\beta_{2}}{2x^2})}}+\frac{(\theta_{1}-1)~e^{-(\frac{\alpha_{1}}{x}+\frac{\beta_{1}}{2x^2})}~(\frac{\alpha_{1}}{x^2}+\frac{\beta_{1}}{x^3})}{1-e^{-(\frac{\alpha_{1}}{x}+\frac{\beta_{1}}{2x^2})}}.\nonumber\\
	\end{eqnarray} 
	From (\ref{eq2.8}), clearly, 
	if $\alpha_{1}= \alpha_{2}=\alpha$, $\beta_{1}= \beta_{2}=\beta$, and  $\theta_{1}\leq \theta_{2}$, then $\frac{d}{dx}\text{log}\Big(\frac{f_{Y}(x)}{f_{X}(x)}\Big)\geq 0$, which proves the desired ordering.
\end{proof}

In order to illustrate Theorem \ref{th2.1}, the following numerical example is considered.
\begin{example}\label{exam2.1}
	Take two IGLFR distributed random variables $X$ and $Y$, such that $X\sim IGLFR(0.5,0.8,2)$ and $Y\sim IGLFR(0.5,0.8,1.2)$. Here, $\theta_1=2>\theta_2=1.2$. Now, we plot the ratio of the PDFs of $X$ and $Y$ in Figure $5,$ which ensures that $X\leq_{lr}Y$ holds.
\end{example}

\begin{figure}[h]
	\centering
	\includegraphics[width=15cm,height=7cm]{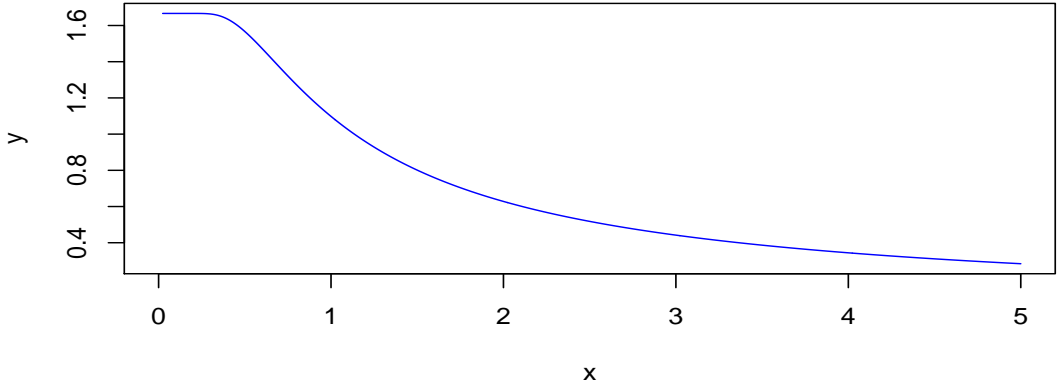}
	\caption{Plot of the ratio of the PDFs of $X$ and $Y$ as in Example \ref{exam2.1}.}
	\label{fig:PdfofIGLFRD**}
\end{figure}

\begin{remark}
	This implication is well-known: $X\leq_{st}Y\Leftarrow X\leq_{rh}Y\Leftarrow X\leq_{lr}Y \Rightarrow X\leq_{hr}Y \Rightarrow X\leq_{st}Y.$ Thus, under similar conditions as in Theorem \ref{th2.1}, the usual stochastic, reversed hazard rate and hazard rate orders also hold. 
\end{remark}

\section{Estimation \setcounter{equation}{0}}\label{sec3}
\subsection{Frequentist approach}

Here, we propose various estimates of parameters of the proposed model. Mainly, we emply two approaches: $(i)$ maximum likelihood estimation and $(ii)$ maximum product spacings estimation. \cite{irshad2021exponentiated} derived MLEs of the parameters of a newly proposed exponentiated unit Lindley distribution. \cite{maya2022unit} obtained MLEs for unit Muth distribution.  First, we derive MLEs of the unknown model parameters for IGLFR distribution.

\subsubsection{Maximum likelihood estimation}
Here, we propose MLEs of the parameters $\alpha$, $\beta$, and $\theta$ for IGLFR distribution. Consider a random sample $(X_{1},\ldots,X_{n})$ drawn from the IGLFR distribution with PDF given by (\ref{eq2.4}). The likelihood function is 
\begin{eqnarray}\label{eq3.1}
L(\alpha,\beta,\theta|data)=\theta^{n}\prod_{i=1}^{n}\left(\frac{\alpha}{x_i^2}+\frac{\beta}{x_i^3}\right)\exp\left\{-\left(\frac{\alpha}{x_i}+\frac{\beta}{2x_i^2}\right)\right\}\left[1-\exp\left\{-\left(\frac{\alpha}{x_i}
+\frac{\beta}{2x_i^2}\right)\right\}\right]^{\theta-1},
\end{eqnarray}
where data=$(x_1,\ldots,x_n)$. The log-likelihood function is obtained as
\begin{eqnarray}\label{eq3.2}
l(\alpha,\beta,\theta)&=&n\log \theta+\sum_{i=1}^{n}\log \left(\frac{\alpha}{x_i^2}+\frac{\beta}{x_i^3}\right)- \sum_{i=1}^{n}\left(\frac{\alpha}{x_i}+\frac{\beta}{2x_i^2}\right)+(\theta-1)\nonumber\\
&~&\times \sum_{i=1}^{n}\log\left[1-\exp\left\{-\left(\frac{\alpha}{x_i}
+\frac{\beta}{2x_i^2}\right)\right\}\right].
\end{eqnarray}
The normal equations can be obtained after differentiating log-likelihood function with respect to the parameters $\alpha$, $\beta$, and $\theta$, and equating them to zero, which are given by 
\begin{eqnarray}
\frac{\partial l}{\partial \alpha} &=&  \sum_{i=1}^{n} \frac{\frac{1}{x_i^2}}{\frac{\alpha}{x_i^2}+\frac{\beta}{x_i^3}} - \sum_{i=1}^{n} \frac{1}{x_i} +(\theta -1) ~ \sum_{i=1}^{n} \frac{\frac{1}{x_i} \exp\{-(\frac{\alpha}{x_i}+\frac{\beta}{2x_i^2})\}}{1-\exp\{-(\frac{\alpha}{x_i}+\frac{\beta}{2x_i^2})\}}=0, \label{eq3.3}\\
\frac{\partial l}{\partial \beta} &=&  \sum_{i=1}^{n} \frac{\frac{1}{x_i^3}}{\frac{\alpha}{x_i^2}+\frac{\beta}{x_i^3}} - \sum_{i=1}^{n} \frac{1}{2x_i^2} +(\theta -1) ~ \sum_{i=1}^{n} \frac{\frac{1}{2x_i^2} \exp\{-(\frac{\alpha}{x_i}+\frac{\beta}{2x_i^2})\}}{1-\exp\{-(\frac{\alpha}{x_i}+\frac{\beta}{2x_i^2})\}}=0,\label{eq3.4}\\
\frac{\partial l}{\partial \theta} &=& \frac{n}{\theta} + \sum_{i=1}^{n} \log\left[1-\exp\left\{-\left(\frac{\alpha}{x_i}+\frac{\beta}{2x_i^2}\right)\right\}\right]=0.\label{eq3.5}
\end{eqnarray}
The above normal equations are difficult to solve simultaneously. Thus, we need numerical technique. The values of $\alpha$ and $\beta$ can be obtained from (\ref{eq3.3}) and (\ref{eq3.4}), using Newton-Raphson method. Then, substituting the values of $\alpha$ and $\beta$ in (\ref{eq3.5}), the value of $\theta$ can be obtained. The MLEs of $\alpha$, $\beta$, and $\theta$ are respectively denoted by $\hat{\alpha}$, $\hat{\beta}$, and $\hat{\theta}.$

\subsubsection{Maximum product spacings estimation}
We derive estimates of the unknown model parameters using maximum product spacings (MPS) estimation method. We recall that the MPS estimation technique is an alternative to MLE for the estimation of parameters of a continuous univariate distribution. The MPS method was introduced by \cite{cheng1979maximum} and \cite{cheng1983estimating}. Later, this method was independently developed as an approximation to the Kullback-Leibler measure by \cite{ranneby1984maximum}. Consider an ordered random sample $X_{1:1}\le\ldots\le X_{n:n}$ from the newly proposed distribution with CDF given by (\ref{eq2.3}). For convenience, in this subsection, we denote $X_{i}=X_{i:n},~i=1,\ldots,n.$ Further, the spacings are denoted by 
\[\text{D}_i = \begin{cases} 
\text{F}(x_1), & \text{when} \ i=1, \\
\text{F}(x_i)-\text{F}(x_{i-1}), & \text{when} \ i=2,\ldots,n, \\
1-\text{F}(x_n), & \text{when} \ i=n+1
\end{cases}
\]
such that $\sum \text{D}_i =1$, see \cite{pyke1965spacings}. The geometric mean of the spacings is given by 
	\begin{eqnarray}\label{eq3.6}
\text{G}^* = \Bigg[\prod_{i=1}^{n+1}\text{D}_i \Bigg]^\frac{1}{n+1}= \Bigg[\text{F}(x_1)\Big\{\prod_{i=2}^{n}[\text{F}(x_i)-\text{F}(x_{i-1})]\Big\} [1-\text{F}(x_n)]\Bigg]^\frac{1}{n+1}.
\end{eqnarray}
Taking logarithm both sides of (\ref{eq3.6}), we obtain 
\begin{eqnarray}\label{eq3.7}
g^*=\log G^* &=&\frac{1}{n+1}  \Bigg[ \log\Big\{1-\big[\gamma(\alpha,\beta,x_1)\big]^ \theta \Big\} + \log\big[\gamma(\alpha,\beta,x_n)\big]^ \theta\nonumber\\
&& + \sum_{i=2}^{n} \log \Big\{ \big[\gamma(\alpha,\beta,x_i)\big]^ \theta -\big[\gamma(\alpha,\beta,x_{i-1})\big]^ \theta  \Big\} \Bigg].
\end{eqnarray}
Now, differentiating (\ref{eq3.7}) with respect to $\alpha,$ $\beta,$ and $\theta$ partially, we get 
\begin{eqnarray}
\frac{\partial g^{*}}{\partial \alpha}&= & \frac{1}{n+1} \Bigg[-\frac{\frac{\theta}{x_1}  \big[1-\gamma(\alpha,\beta,x_1)\big] \big[\gamma(\alpha,\beta,x_1)\big]^ {\theta-1} }{1-\big[\gamma(\alpha,\beta,x_1)\big]^ \theta } 
+ \frac{\frac{\theta}{x_n} [1-\gamma(\alpha,\beta,x_n)] \big[\gamma(\alpha,\beta,x_n)\big]^ {\theta-1} }{\big[\gamma(\alpha,\beta,x_n)\big]^ \theta }
\nonumber\\
&~&+\sum_{i=2}^{n} \frac{\frac{\theta}{x_i} [1-\gamma(\alpha,\beta,x_i)] \big[\gamma(\alpha,\beta,x_i)\big]^ {\theta-1} - \frac{\theta}{x_{i-1}}  [1-\gamma(\alpha,\beta,x_{i-1})] \big[\gamma(\alpha,\beta,x_{i-1})\big]^ {\theta-1} }{\big[\gamma(\alpha,\beta,x_{i-1})\big]^ \theta -\big[\gamma(\alpha,\beta,x_i)\big]^ \theta  } \Bigg],\nonumber\\
\end{eqnarray}

\begin{eqnarray}
\frac{\partial g^*}{\partial \beta} &= & \frac{1}{n+1} \Bigg[-\frac{\frac{\theta}{2x_1^2}  [1-\gamma(\alpha,\beta,x_1)] \big[\gamma(\alpha,\beta,x_1)\big]^ {\theta-1} }{1-\big[\gamma(\alpha,\beta,x_1)\big]^ \theta } +
 \frac{\frac{\theta}{2x_n^2} [1-\gamma(\alpha,\beta,x_n)] \big[\gamma(\alpha,\beta,x_n)\big]^ {\theta-1} }{\big[\gamma(\alpha,\beta,x_n)\big]^ \theta }\nonumber\\
&&+\sum_{i=2}^{n} \frac{\frac{\theta}{2x_i^2}  [1-\gamma(\alpha,\beta,x_i)] \big[\gamma(\alpha,\beta,x_i)\big]^ {\theta-1} - \frac{\theta}{2x_{i-1}^2} [1-\gamma(\alpha,\beta,x_{i-1})] \big[\gamma(\alpha,\beta,x_{i-1})\big]^ {\theta-1} }{\big[\gamma(\alpha,\beta,x_{i-1})\big]^ \theta -\big[\gamma(\alpha,\beta,x_i)\big]^ \theta  } \Bigg],\nonumber\\
\end{eqnarray}
and
\begin{eqnarray}
\frac{\partial g^*}{\partial \theta} &= & \frac{1}{n+1} \Bigg[\frac{-\big[\gamma(\alpha,\beta,x_1)\big]^{\theta} \log \big[\gamma(\alpha,\beta,x_1)\big] }{1-\big[\gamma(\alpha,\beta,x_1)\big]^ \theta } + \frac{\big[\gamma(\alpha,\beta,x_{n})\big]^{\theta} \log \big[\gamma(\alpha,\beta,x_n)\big] }{\big[\gamma(\alpha,\beta,x_n)\big]^ \theta }\nonumber\\ 
&~&+\sum_{i=2}^{n} \frac{\big[\gamma(\alpha,\beta,x_{i-1})\big]^{\theta} \log \big[\gamma(\alpha,\beta,x_{i-1})\big]-\big[\gamma(\alpha,\beta,x_i)\big]^{\theta} \log \big[\gamma(\alpha,\beta,x_i)\big]}{ \big[\gamma(\alpha,\beta,x_{i-1})\big]^ \theta -\big[\gamma(\alpha,\beta,x_i)\big]^ \theta }  \Bigg].
\end{eqnarray}

Clearly, no explicit solutions can be obtained from the nonlinear equations $\frac{\partial g^{*}}{\partial \alpha}=0$, $\frac{\partial g^{*}}{\partial \beta}=0$, and $\frac{\partial g^{*}}{\partial \theta}=0$. Thus, it is required to adopt a suitable numerical technique, say Newton-Raphson method. The solution for $\alpha,$ $\beta,$ and $\theta$ is the MPS estimates. 

\subsubsection{Asymptotic confidence intervals}
Here, we derive asymptotic confidence intervals of  $\alpha,$ $\beta,$ and $\theta$ of the IGLFR distribution. Here, the MLEs of the parameters can not be obtained explicitely. Thus, it is impossible to obtain their exact distributions. Under some regularity conditions, the MLEs $(\hat{\alpha},\hat{\beta},\hat{\theta})$ asymptotically follow tri-variate normal distribution with mean vector $(\alpha,\beta,\theta)$ and variance-covariance matrix $\text{I}_0^{-1}$, where 
\begin{equation}\nonumber
\begin{aligned} 
\text{I}_0^{-1} = &
\begin{pmatrix}
-\frac{\partial^2 \log \text{L}}{\partial \alpha^2} & -\frac{\partial^2 \log \text{L}}{\partial \alpha \partial \beta} & -\frac{\partial^2 \log \text{L}}{\partial \alpha \partial \theta}\\ \\
-\frac{\partial^2 \log \text{L}}{\partial \alpha \partial \beta } & -\frac{\partial^2 \log \text{L}}{\partial \beta^2} & -\frac{\partial^2 \log \text{L}}{\partial \beta \partial \theta} \\
\\
-\frac{\partial^2 \log \text{L}}{\partial \theta \partial \alpha } & -\frac{\partial^2 \log \text{L}}{\partial \theta \partial \beta } & -\frac{\partial^2 \log \text{L}}{\partial \theta^2 }
\end{pmatrix}^{-1}
= & \begin{pmatrix}
\text{Var}(\hat{\alpha}) & \text{Cov}(\hat{\alpha},\hat{\beta}) & \text{Cov}(\hat{\alpha},\hat{\theta})\\ \\
\text{Cov}(\hat{\alpha},\hat{\beta}) & \text{Var}(\hat{\beta}) & \text{Cov}(\hat{\beta},\hat{\theta}) \\ \\
\text{Cov}(\hat{\alpha},\hat{\theta}) & \text{Cov}(\hat{\beta},\hat{\theta}) & \text{Var}(\hat{\theta})
\end{pmatrix}
\end{aligned}
\end{equation}
and 
\begin{equation}\nonumber
\begin{aligned}
\nonumber\frac{\partial^2 \log \text{L}}{\partial \alpha^2}= & - \sum_{i=1}^{n} \frac{\frac{1}{x_i^4}}{\big(\frac{\alpha}{x_i^2}+\frac{\beta}{x_i^3}\big)^2} - (\theta -1) ~ \sum_{i=1}^{n} \frac{e^{-(\frac{\alpha}{x_i}+\frac{\beta}{2x_i^2})}\frac{1}{x_i^2}}{\big[1-e^{-(\frac{\alpha}{x_i}+\frac{\beta}{2x_i^2})}\big]^2},\\ 
\frac{\partial^2 \log \text{L}}{\partial \alpha \partial \beta}= &  - \sum_{i=1}^{n} \frac{\frac{1}{x_i^5}}{\big(\frac{\alpha}{x_i^2}+\frac{\beta}{x_i^3}\big)^2} - (\theta -1) ~ \sum_{i=1}^{n} \frac{e^{-(\frac{\alpha}{x_i}+\frac{\beta}{2x_i^2})}\frac{1}{2x_i^3}}{\big[1-e^{-(\frac{\alpha}{x_i}+\frac{\beta}{2x_i^2})}\big]^2},\\ 
\nonumber \frac{\partial^2 \log \text{L}}{\partial \alpha \partial \theta}= & \sum_{i=1}^{n} \frac{e^{-(\frac{\alpha}{x_i}+\frac{\beta}{2x_i^2})}\frac{1}{x_i}}{\big[1-e^{-(\frac{\alpha}{x_i}+\frac{\beta}{2x_i^2})}\big]},~~~~
\nonumber\frac{\partial^2 \log \text{L}}{\partial \beta \partial \theta}=  \sum_{i=1}^{n} \frac{e^{-(\frac{\alpha}{x_i}+\frac{\beta}{2x_i^2})}\frac{1}{2x_i^2}}{\big[1-e^{-(\frac{\alpha}{x_i}+\frac{\beta}{2x_i^2})}\big]},
\\
\nonumber\frac{\partial^2 \log \text{L}}{\partial \beta^2}= & - \sum_{i=1}^{n} \frac{\frac{1}{x_i^6}}{\big(\frac{\alpha}{x_i^2}+\frac{\beta}{x_i^3}\big)^2} - (\theta -1) ~ \sum_{i=1}^{n} \frac{e^{-(\frac{\alpha}{x_i}+\frac{\beta}{2x_i^2})}\frac{1}{4x_i^4}}{\big[1-e^{-(\frac{\alpha}{x_i}+\frac{\beta}{2x_i^2})}\big]^2},~~
\nonumber \frac{\partial^2 \log \text{L}}{ \partial \theta^2}= -\frac{n}{\theta^2}.
\end{aligned}
\end{equation}

Thus, the $100(1-\nu)$\% confidence intervals for the parameters $\alpha$, $\beta$, and $\theta$ are  respectively obtained as
\begin{equation}\nonumber
\hat{\alpha} \pm \text{Z}_ \frac{\nu}{2} \sqrt{\text{Var} (\hat\alpha)} ,~ ~
\hat{\beta} \pm \text{Z}_ \frac{\nu}{2}  \sqrt{\text{Var}(\hat\beta)} ,~ ~
\hat{\theta} \pm \text{Z}_ \frac{\nu}{2}  \sqrt{\text{Var}(\hat\theta)} ,
\end{equation}
where $Z_\frac{\nu}{2}$ is the upper $\frac{\nu}{2}$th percentile of the standard normal distribution.

\subsection{Bayesian approach}
In many real life applications, the parameters of a model can not be fixed. It varies and as a result, the model parameters have some randomness. Bayesian analysis include the randomness of the parameter into study. Bayesian approach allows us to consider the prior beliefs (in the form of prior distribution) about the  parameter with the available data information (likelihood function). We note that if proper information about the unknown parameters are available, then proper prior distributions are taken into account, otherwise, improper prior distributions are considered in the study. In this subsection, we will derive Bayes estimates of $\alpha$, $\beta$, and $\theta$ with respect to SELF. Let $\delta$ be an estimator of the parameter $\lambda.$ Then, the SELF is defined as 
\begin{eqnarray}
L(\lambda,\delta)=(\delta-\lambda)^2.
\end{eqnarray}
We recall that the SELF is symmetric. Further, under SELF, the mean of the posterior distribution is a Bayes estimator. Note that here the conjugate priors do not exist. Thus, according to \cite{kundu2009bayesian}, we assume independent gamma priors for the unknown parameters. A random variable $T$ is said to follow gamma distribution with parameters $m$ and $n$ if then PDF is 
\begin{eqnarray}\label{eq3.12}
f^*(t|m,n)=\frac{1}{\Gamma(m)n^m}t^{m-1}\exp\{-\frac{t}{n}\},~t>0,~m,n>0.
\end{eqnarray}
Henceforth, we denote $T\sim Gamma(m,n)$ if $T$ has the PDF given by  (\ref{eq3.12}). Consider $\alpha \sim Gamma(a,b)$, $\beta \sim Gamma(c,d),$ and $\theta \sim Gamma(p,q)$. After some calculations, the joint posterior distribution of $\alpha,~\beta,$ and $\theta$ is obtained as  
 \begin{equation}\label{eq3.13}
\pi(\alpha,\beta,\theta |data)= \frac{k  \prod_{i=1}^{n}\Big[{(\frac{\alpha}{x_i^2}+\frac{\beta}{x_i^3})} ~ {[1-\gamma(\alpha,\beta,x_i)]} ~ \big[\gamma(\alpha,\beta,x_i)\big]^{\theta-1}\Big]}{\int_{0}^{\infty}\int_{0}^{\infty}\int_{0}^{\infty}  k \prod_{i=1}^{n}\Big[{(\frac{\alpha}{x_i^2}+\frac{\beta}{x_i^3})} ~ {[1-\gamma(\alpha,\beta,x_i)]} ~ \big[\gamma(\alpha,\beta,x_i)\big]^{\theta-1}\Big] \ d\alpha \ d\beta \  d\theta},
\end{equation}
where $k=k(\alpha,\beta,\theta)=\frac{1}{\Gamma(a) b^a} \frac{1}{\Gamma(c) d^c}\frac{1}{\Gamma(p) q^p}\alpha^{a-1}\beta^{c-1}\theta^{p-1+n} \exp\{-\big(\frac{\alpha}{b}+\frac{\beta}{d}+\frac{\theta}{q}\big)\}.$ Thus, the Bayes estimate of $\alpha$ with respect to SELF, denoted by $\hat{\alpha}_{\text{BE}}$, is obtained as 
\begin{eqnarray}\label{eq2.14}
\hat{\alpha}_{\text{BE}}&=& \int_{0}^{\infty}\int_{0}^{\infty}\int_{0}^{\infty} \alpha ~ \pi(\alpha,\beta,\theta |data)\ d\alpha \ d\beta \  d\theta,
\nonumber\\&=& \frac{\int_{0}^{\infty}\int_{0}^{\infty}\int_{0}^{\infty}\alpha  k \prod_{i=1}^{n}\Big[{(\frac{\alpha}{x_i^2}+\frac{\beta}{x_i^3})} ~ {[1-\gamma(\alpha,\beta,x_i)]}  \big[\gamma(\alpha,\beta,x_i)\big]^{\theta-1}
	\Big]\ d\alpha \ d\beta \  d\theta}{\int_{0}^{\infty}\int_{0}^{\infty}\int_{0}^{\infty}  \text{k} ~ \prod_{i=1}^{n}\Big[{(\frac{\alpha}{x_i^2}+\frac{\beta}{x_i^3})} ~ {[1-\gamma(\alpha,\beta,x_i)]}~ \Big[\gamma(\alpha,\beta,x_i)\Big]^{\theta-1}\Big] \ d\alpha \ d\beta \  d\theta},\nonumber\\
&=& \eta_{1}(data).
\end{eqnarray}
The Bayes estimates of $\beta$ and $\theta$, respectively denoted by $\hat{\beta}_{\text{BE}}$ and $\hat{\theta}_{\text{BE}}$ can be obtained similarly. Here, we omit the expressions for $\hat{\beta}_{\text{BE}}$ and $\hat{\theta}_{\text{BE}}$ to avoid repetitions.  We observe that it is impossible to obtain the Bayes estimates of $\alpha$, $\beta,$ and $\theta$ in closed form. Thus, we will employ MCMC technique to compute the approximate Bayes estimates, which is discussed in the next subsection. 

\subsubsection{MCMC Method}
Here, MCMC approach will be employed to compute the approximate Bayes estimates of the parameters under SELF. From the posterior density function given by (\ref{eq3.13}), we obtain the following conditional posterior density functions: 

\begin{equation}
\pi_{1}(\alpha|\beta,\theta ,data) \propto ~ \alpha^{a-1} \exp\{-\frac{\alpha}{b}\} \prod_{i=1}^{n}\Bigg[{\Big(\frac{\alpha}{x_i^2}+\frac{\beta}{x_i^3}\Big)} ~ {\exp\{-\frac{\alpha}{x_i}\}} ~ \Big[\gamma(\alpha,\beta,x_i)\Big]^{\theta-1}\Bigg],
\end{equation}
\begin{equation}
\pi_{2}(\beta|\alpha,\theta ,data) \propto ~ \beta^{c-1} \exp\{-\frac{\beta}{d}\} \prod_{i=1}^{n}\Bigg[{\Big(\frac{\alpha}{x_i^2}+\frac{\beta}{x_i^3}\Big)} ~ {\exp\{-\frac{\beta}{2x_i^2}\}} ~ \Big[\gamma(\alpha,\beta,x_i)\Big]^{\theta-1}\Bigg]
\end{equation}
and
\begin{equation}
\pi_{3}(\theta|\alpha,\beta,data) \propto ~ \theta^{p-1+n} \exp\{-\frac{\theta}{q}\} \prod_{i=1}^{n} \Big[\gamma(\alpha,\beta,x_i)\Big]^{\theta-1}.
\end{equation} 
The above density functions  $\pi_{1}(\alpha|\beta,\theta ,data), \pi_{2}(\beta|\alpha,\theta,data),$ and  $\pi_{3}(\theta|\alpha,\beta,data)$ can
not be written in the form of any well known distributions. Thus, the MCMC samples can not be generated from these densities.  So, the
Metropolis-Hastings algorithm is utilized to obtain MCMC samples from the conditional density functions. We obtain the Bayes estimates by using the following steps: \\ \\
\textbf{Step 1:} Choose initial values as $\alpha^{(1)} = \hat\alpha$, $\beta^{(1)} = \hat\beta,$ and $\theta^{(1)} = \hat\theta$. Set $i = 1$.\\ 
\textbf{Step 2:} Generate $\alpha^{(i)},~ \beta^{(i)},$ and $ \theta^{(i)}$ with normal distribution as $\alpha^{(i)}$ $\thicksim$ N $(\alpha^{(i-1)}$, $Var(\hat\alpha))$, $\beta^{(i)}$ $\thicksim$ N$(\beta^{(i-1)}$,Var($\hat\beta))$, and $\theta^{(i)}$ $\thicksim$ N$(\theta^{(i-1)}$,Var($\hat\theta))$.\\ 
\textbf{Step 3:} Compute $\Omega_{\alpha} =min \Big(1,\frac{\pi_{1}(\alpha^{(i)}|\beta^{(i-1)},\theta^{(i-1)} ,data)}{\pi_{1}(\alpha^{(i-1)}|\beta^{(i-1)},\theta^{(i-1)} ,data)}\Big)$, $\Omega_{\beta} =min \Big(1,\frac{\pi_{2}(\beta^{(i)}|\alpha^{(i-1)},\theta^{(i-1)} ,data)}{\pi_{2}(\beta^{(i-1)}|\alpha^{(i-1)},\theta^{(i-1)} ,data)}\Big),$ and  $\Omega_{\theta} =min \Big(1,\frac{\pi_{3}(\theta^{(i)}|\alpha^{(i-1)},\beta^{(i-1)} ,data)}{\pi_{3}(\theta^{(i-1)}|\alpha^{(i-1)},\beta^{(i-1)} ,data)}\Big)$. \\ \\
\textbf{Step 4:} Generate samples for $\mathcal{T}_{1}$ $\thicksim$ Uniform(0,1),
$\mathcal{T}_{2}$ $\thicksim$ Uniform(0,1), and $\mathcal{T}_{3}$ $\thicksim$ Uniform(0,1).\\ \\
\textbf{Step 5:} Set
\[\alpha = \begin{cases} 
\alpha^{(i)}, & \text{if} \ \mathcal{T}_{1}\le \Omega_{\alpha}, \\
\alpha^{(i-1)}, & \text{otherwise}. \\
\end{cases}
\]
\[\beta = \begin{cases} 
\beta^{(i)}, & \text{if} \ \mathcal{T}_{2}\le \Omega_{\beta}, \\
\beta^{(i-1)}, & \text{otherwise}. \\
\end{cases}
\]
\[\theta = \begin{cases} 
\theta^{(i)}, & \text{if} \ \mathcal{T}_{3}\le \Omega_{\theta}, \\
\theta^{(i-1)}, & \text{otherwise}. \\
\end{cases}
\]
\\ \textbf{Step 6:} Set i=i+1. \\ \\
\textbf{Step 7:} Repeat steps 1 to 6, $K$ times to get $\alpha^{(1)}, \cdots, \alpha^{(K)}$ ; $\beta^{(1)}, \cdots, \beta^{(K)}$ and $\theta^{(1)} \cdots, \theta^{(K)}$.\\

Thus, under SELF, the Bayes estimates of $\alpha, \beta,$ and $\theta$  are respectively given as
$$\hat\alpha_{SEL}=\frac{1}{K}~\sum_{i=1}^{K}\alpha^{(i)}, ~ \hat\beta_{SEL}=\frac{1}{K}~\sum_{i=1}^{K}\beta^{(i)},~ \mbox{and}~ \hat\theta_{SEL}=\frac{1}{K}~\sum_{i=1}^{K}\theta^{(i)}.$$

The MCMC method is also used to derive $100(1-\frac{\nu}{2})$\% credible intervals of the parameters $\alpha$, $\beta$, and $\theta$ as 
$$\Big({\alpha}_{(t)},~ {\alpha}_{\big(t+(1-\nu)K\big)}\Big),~ 
\Big({\beta}_{(t)},~ {\beta}_{\big(t+(1-\nu)K\big)}\Big),~ \mbox{and}~
\Big({\theta}_{(t)},~ {\theta}_{\big(t+(1-\nu)K\big)}\Big).$$

\section{Simulation study \setcounter{equation}{0}}\label{sec4}
Here, a Monte Carlo simulation is executed  to see the comparative performance of the proposed estimators. On the basis of average bias and mean squared error (MSE), we will examine the estimators' performance. Sample sizes are considered as  $n=20,~30,~40,~50,~60,~70,~80,$ and $90.$ $10,000$ independent samples of size $n$ are generated from IGLFR distribution with two sets of parameters $\alpha=0.5$, $\beta=0.5$, $\theta=1$ and $\alpha=0.3$, $\beta=0.5$, $\theta=0.2$.  We report average bias, obtained by MLE, MPS and Bayes estimates along with their MSEs presented in the parentheses in Table $1$, Table $2$, and Table $3$, for different values of $\alpha,~\beta,$ and $\theta$. All simulations have done by using $R$ software. To compute MLE and MPSE, 'nleqslv' package is used. To obtain Bayes estimates, MCMC samples have been generated by using 'coda' package. The $95\%$ asymptotic confidence and Bayesian credible intervals along with interval lengths (ALs) and coverage probabilities (CPs) are presented in Tables $4$, $5$, and $6$. From the tables, following observations are noticed.
\begin{itemize}
\item From Table $1$, Table $2$, and Table $3$, we clearly observe that the mean squared errors and average biases decrease as the sample size increases. 
\item From Table $4$, Table $5$, and Table $6$, we notice that when the sample size increases the average confidence/credible lengths decrease. There is no specific pattern of increasing or decreasing values of CPs when the sample size increase. 
\item For small sample sizes, the Bayesian credible intervals or the asymptotic confidence intervals are slightly skewed, and they become symmetric for large sample sizes. 
\item Based on average bias and MSEs, in general, the Bayes estimates give better results than the MLEs and MPS estimates.
\item Based on average bias and MSEs, in general, the MLEs give better performance than the MPS estimates.
\item The Bayesian credible intervals provide superior result than the asymptotic confidence intervals based on the average interval lengths.
\item Based on CP, BCIs perform better than ACIs. 
\end{itemize}

From above observations, we can summarize that Bayes estimators have better performance than the classical estimators. In classical estimators, one can prefer MPSE over MLE. Bayesian estimation contains more information than the maximum likelihood estimation. In terms of AL and CP of interval estimates, an experimenter can choose BCIs over ACIs. 

 \begin{table}[htbp!]
	\begin{center}
		\caption{Average bias and MSEs (shown in parentheses) of the estimates for the IGLFR distribution with $\alpha=0.5$, $\beta=0.5$, and $\theta=1$ with different values of $n$.}
		\label{T1}
		\tabcolsep 6pt
		\small
		\scalebox{0.85}{
			\begin{tabular}{*{10}c*{9}{r@{}l}}
				\toprule
				\multicolumn{1}{c}{} &
				\multicolumn{3}{c}{MLE} & \multicolumn{3}{c}{MPSE} & \multicolumn{3}{c}{Bayes} \\
				\cmidrule(lr){2-4}\cmidrule(lr){5-7} \cmidrule(lr){8-10}  
				\multicolumn{1}{c}{$n$}& \multicolumn{1}{c}{$\alpha$} & \multicolumn{1}{c}{$\beta$} & \multicolumn{1}{c}{$\theta$}& \multicolumn{1}{c}{$\alpha$} & \multicolumn{1}{c}{$\beta$} & \multicolumn{1}{c}{$\theta$}& \multicolumn{1}{c}{$\alpha$} & \multicolumn{1}{c}{$\beta$} & \multicolumn{1}{c}{$\theta$} \\
				\hline
				20& 0.1529& -0.0121& 0.2827& 0.4100& -0.2377&  0.6820& 0.0019& 0.0310& 0.0268 \\  
				& (0.2746)& (0.1470)& (0.6573) & (0.3628)& (0.1707)& (1.1797)&(0.0290)&(0.0408)&(0.0545)  \\
				30& 01090 &0.0034 &0.1822 &0.3788 & -0.2327 &0.5338 &-0.0037  &0.0297 &0.0240\\
				&(0.1898) &(0.0959) &(0.3146) &(0.3227) &(0.1298) &(0.6972) &(0.0286) &(0.0360) &(0.0478) \\
				40 &0.0896 &-0.0181 &0.1273 &0.3277 &-0.2177 &0.4110 &0.0011 &0.0132 &0.0143 \\
				&(0.1673) &(0.0792) &(0.2163) &(0.2601) &(0.1147) &(0.4390) &(0.0290) &(0.0283) &(0.0376)\\
				50 &0.0437 &0.0032 &0.0764 &0.2763 &-0.1833 &0.3396 &-0.0064 &0.0180 &0.0089\\
				&(0.1378) &(0.0673) &(0.1594) &(0.2173) &(0.0953) &(0.3387) &(0.0289) &(0.261) &(0.0373)\\
				60 &0.0661 &-0.0207 &0.0894 &0.2761 &-0.1898 &0.3150 &0.0022 &0.0055 &0.0162\\ 
				&(0.1244) &(0.0565) &(0.1393) &(0.1971) &(0.0895) &(0.2638) &(0.0288) &(0.0219) &(0.0341)\\
				70 &0.0502 &-0.0084 &0.0733 &0.2513 &-0.1681 &0.2862 &-0.0049 &0.0157 &0.0109\\
				&(0.1145) &(0.0511) &(0.1190) &(0.1843) &(0.0813) &(0.2356) &(0.0268) &(0.0224) &(0.0305)\\
				80 &0.0399 &-0.0066 &0.0580 &0.2230 &-0.1509 &0.2484 &-0.0031 &0.0101 &0.0103\\
				&(0.1084) &(0.0510) &(0.0995) &(0.1649) &(0.0775) &(0.1936) &(0.0288) &(0.0216) &(0.0297)\\
				90 &0.0259 &-0.0053 &0.0389 &0.2105 &-0.1471 &0.2233 &-0.0066 &0.0083 &0.0024\\
				&(0.0931) &(0.0430) &(0.0866) &(0.1392) &(0.0656) &(0.1553) &(0.0255) &(0.0181) &(0.0276)\\
				100 &0.0300 &-0.0083 &0.0386 &0.2034 &-0.1418 &0.2091 &-0.0066 &0.0096 &0.0011\\ &(0.0859) &(0.0384) &(0.0770) &(0.1299) &(0.0594) &(0.1381) &(0.0248) &(0.0174) &(0.0247)\\
				150 &0.0004 &0.0063 &0.0180 &0.1426 &-0.0990 &0.1446 &-0.0185 &0.0142 &-0.0082\\
				&(0.0631) &(0.0284) &(0.0498) &(0.0819) &(0.0386) &(0.0745) &(0.0238) &(0.0150) &(0.0196)\\
				\bottomrule
		\end{tabular}}
	\end{center}
	\vspace{-0.5cm}
\end{table}
\begin{table}[htbp!]
	\begin{center}
		\caption{Average bias and MSEs (shown in parentheses) of the estimates for the IGLFR distribution with $\alpha=0.3$, $\beta=0.5$, and $\theta=0.2$ with different values of $n$.}
		\label{T2}
		\tabcolsep 6pt
		\small
		\scalebox{0.85}{
			\begin{tabular}{*{11}c*{10}{r@{}l}}
				\toprule
				\multicolumn{1}{c}{} &
				\multicolumn{3}{c}{MLE} & \multicolumn{3}{c}{MPSE} & \multicolumn{3}{c}{Bayes} \\
				\cmidrule(lr){2-4}\cmidrule(lr){5-7} \cmidrule(lr){8-10}  
				\multicolumn{1}{c}{$n$}& \multicolumn{1}{c}{$\alpha$} & \multicolumn{1}{c}{$\beta$} & \multicolumn{1}{c}{$\theta$}& \multicolumn{1}{c}{$\alpha$} & \multicolumn{1}{c}{$\beta$} & \multicolumn{1}{c}{$\theta$}& \multicolumn{1}{c}{$\alpha$} & \multicolumn{1}{c}{$\beta$} & \multicolumn{1}{c}{$\theta$} \\
				\midrule
				20 &0.1389 &0.2113 &0.0179 &0.2014 & -0.1460 &0.0521 &-0.0568 &-0.0448 &-1.314\\  
				&(0.2283) &(0.6509) &(0.0056) &(0.2170) &(0.2477) &(0.0970) &(0.0535) &(0.0579) &(3.1959)  \\ 
				30 &0.0885 &0.1094 &0.0169 &0.1615 &-0.1386 &0.0363 &-0.0491 &-0.0463 &-1.3032\\
				&(0.1313) &(0.3323) &(0.0255) &(0.1364) &(0.1986) &(0.0053) &(0.0497) &(0.0550) &(3.0672)\\
				40 &0.0937 &0.0771 &0.0083 &0.1679 &-0.1355 &0.0270 &-0.0603 &-0.0370 &-1.3147\\
				&(0.1089) &(0.2033) &(0.0039) &(0.1256) &(0.1491) &(0.0032) &(0.0517) &(0.0535) &(2.8042)\\
				50 &0.0387 &0.0659 &0.0068 &0.1138 &-0.1168 &0.0196 &-0.0470& -0.0305 &-1.1898\\
				&(0.0726) &(0.0256) &(0.0130) &(0.0861) &(0.0975) &(0.0023) &(0.0484) &(0.0498) &(2.373)\\
				60 &0.0398 &0.0396 &0.0099 &0.1163 &-0.1210 &0.0182 &-0.0446 &-0.0246 &-1.1871\\
				&(0.0632) &(0.9200) &(0.0288) &(0.0813) &(0.0854) &(0.0019) &(0.0437)&(0.0476) &(2.3030)\\
				70 &0.0501 &0.0209 &0.0091 &0.1241 &-0.1173 &0.0192 &-0.0349 &-0.0340 &-1.0715\\
				&(0.0686) &(0.0973) &(0.0615) &(0.0867) &(0.0878) &(0.0019) &(0.0446) &(0.0426) &(1.9222)\\
				80 &0.0274 &0.0366 &0.0063 &0.0898 &-0.0919 &0.0154 &-0.0345 &-0.0291 &-1.006\\
				&(0.0533) &(0.0817) &(0.0078) &(0.0655) &(0.0713) &(0.0014) &(0.0426) &(0.0425) &(1.6935)\\
				90 &0.0290 &0.0193 &0.0060 &0.0914 &-0.0989 &0.0127 &-0.0331 &-0.0252 &-1.0627\\
				&(0.0471) &(0.0796) &(0.0332) &(0.0573) &(0.0547) &(0.0011) &(0.0404) &(0.0391) &(1.7514)\\
				100 &0.0286 &0.0089 &0.0058 &0.0961 &-0.1045 &0.0129 &-0.0443 &-0.0173 &-0.9954\\ &(0.0455) &(0.0789 ) &(0.0186) &(0.0574) &(0.0544) &(0.0011) &(0.0361) &(0.0362) &(1.5309)\\
				150 &0.0043 &0.0141 &-0.0009 &0.0616 &-0.0627 &0.0101 &-0.0055 &-0.0015 &-0.0808 \\ 
				&(0.0321) &(0.0467) &(0.0073) &(0.0386) &(0.0363) &(0.0024) &(0.0016) &(0.0012) &(0.0995)\\
				\bottomrule
		\end{tabular}}
	\end{center}
	\vspace{-0.5cm}
\end{table}
\begin{table}[htbp!]
	\begin{center}
		\caption{Average bias and MSEs (shown in parentheses) of the estimates for the IGLFR distribution  with $\alpha=1$, $\beta=0.7$, and $\theta=0.4$ with different values of $n$.}
		\label{T3}
		\tabcolsep 6pt
		\small
		\scalebox{0.85}{
			\begin{tabular}{*{11}c*{10}{r@{}l}}
				\toprule
				\multicolumn{1}{c}{} &
				\multicolumn{3}{c}{MLE} & \multicolumn{3}{c}{MPSE} & \multicolumn{3}{c}{Bayes} \\
				\cmidrule(lr){2-4}\cmidrule(lr){5-7} \cmidrule(lr){8-10}  
				\multicolumn{1}{c}{$n$}& \multicolumn{1}{c}{$\alpha$} & \multicolumn{1}{c}{$\beta$} & \multicolumn{1}{c}{$\theta$}& \multicolumn{1}{c}{$\alpha$} & \multicolumn{1}{c}{$\beta$} & \multicolumn{1}{c}{$\theta$}& \multicolumn{1}{c}{$\alpha$} & \multicolumn{1}{c}{$\beta$} & \multicolumn{1}{c}{$\theta$} \\
				\midrule
				20 &0.1299 &0.3016 &0.1220 &0.1079 &0.1345 &0.0974 &-0.0978 &-0.0852 &-0.0869 \\
				&(0.6000) &(2.5044) &(0.0224) &(0.4964) &(0.9482) &(0.0360) &(0.1088) &(0.0864) &(0.3562)\\
				30 &0.1224 &0.2930 &0.1122 &0.0966 &0.1248 &0.0735 &-0.0957 &-0.0740&-0.0270  \\
				&(0.3851) &(1.2016) &(0.0144) &(0.3429) &(0.5968) &(0.0240) &(0.0922) &(0.0781) &(0.1109)\\
				40 &0.1108 &0.2688 &0.1046 &0.0952 &0.1195 &0.0505 &-0.0891 &-0.0687 &-0.0221\\
				&(0.3243) &(0.7481) &(0.0103) &(0.2905) &(0.4606) &(0.0138) &(0.0820) &(0.0691) &(0.0731)\\
				50 &0.1044 &0.2421 &0.1040 &0.0937 &0.1074 &0.0449 &-0.0784 &-0.0500 &-0.0165\\
				&(0.2685) &(0.6839) &(0.0086) &(0.2346) &(0.4379) &(0.0113) &(0.0730) &(0.0709) &(0.0670)\\
				60 &0.0978 &0.2399 &0.1004 &0.0927 &0.1038 &0.0356 &-0.0705 &-0.0397 &-0.0093\\
				&(0.2301) &(0.5042) &(0.0061) &(0.2147) &(0.3523) &(0.0078) &(0.0717) &(0.0731) &(0.0380)\\
				70 &0.0908 &0.2664 &0.0987 &0.0894 &0.0983 &0.0350 &-0.0600 &-0.0509 &-0.0069\\
				&(0.2125) &(0.4120) &(0.0060) &(0.1951) &(0.3147) &(0.0073) &(0.0601) &(0.0630) &(0.0270)\\
				80 &0.0897 &0.2364 &0.0941 &0.0875 &0.0958 &0.0289 &-0.0402 &-0.0385 &-0.0064\\
				&(0.1820) &(0.3720) &(0.0052 ) &(0.1797) &(0.2942) &(0.0065) &(0.0623) &(0.0709) &(0.0287)\\
				90 &0.0855 &0.2226 &0.0925 &0.0841 &0.0918 &0.0289 &-0.0402 &-0.0385 &-0.0064\\
				&(0.1728) &(0.3379) &(0.0045) &(0.1752) &(0.2802) &(0.0055) &(0.0539) &(0.0628) &(0.0239)\\
				100 &0.0839 &0.2016 &0.0888 &0.0813 &0.0909 &0.0269 &-0.0398 &-0.0341 &-0.0016\\
				&(0.1542) &(0.2763) &(0.0039) &(0.1617) &(0.2512) &(0.0048) &(0.0545) &(0.0580) &(0.0106)\\
				150 &0.0828 &0.1937 &0.0832 &0.0805 &0.0891 &0.0207 &-0.0292 &-0.0099 &0.0048\\
				&(0.1016) &(0.1903) &(0.0023) &(0.1103) &(0.1786) &(0.0029) &(0.0416) &(0.0507) &(0.0016)\\
				\bottomrule
		\end{tabular}}
	\end{center}
	\vspace{-0.5cm}
\end{table}

\begin{table}[htbp!]
	\begin{center}
		\caption{ ALs of $95\%$ asymptotic confidence intervals (ACI) and Bayesian credible intervals (BCI) and their corresponding CPs (in parentheses) with $\alpha=0.5$, $\beta=0.5$, and $\theta=1$ with different values of $n$. }
		\label{T4}
		\tabcolsep 2pt
		\small
		\scalebox{0.85}{
			\begin{tabular}{*{11}c*{10}{r@{}l}}
				\toprule
				\multicolumn{1}{c}{} &
				\multicolumn{3}{c}{ACI} & \multicolumn{3}{c}{BCI} \\
				\cmidrule(lr){2-4}\cmidrule(lr){5-7}  
				\multicolumn{1}{c}{$n$}& \multicolumn{1}{c}{$\alpha$} & \multicolumn{1}{c}{$\beta$} & \multicolumn{1}{c}{$\theta$}& \multicolumn{1}{c}{$\alpha$} & \multicolumn{1}{c}{$\beta$} & \multicolumn{1}{c}{$\theta$} \\
				\midrule
				20& 1.9150 & 1.4573 & 2.7331 & 0.6034 & 0.7079 & 0.8788\\
				& (0.9342)& (0.9298)& (0.9391)& (0.9485)& (0.9507)& (0.9527) \\
				30& 1.6524 &1.267 &2.1736 &0.6033 &0.6611 & 0.7868 \\
				& (0.9358)& (0.9314)& (0.9402)& (0.9473)& (0.9518)& (0.9511) \\
				40& 0.14596 &1.1098 & 1.7382 & 0.6099 &0.6210 &0.7239\\
				& (0.9414)& (0.9337)& (0.9430)& (0.9510)& (0.9522)& (0.9517) \\
				50& 1.3370 &1.0681 &1.5426 &0.6226 & 0.604 &0.7239\\
				& (0.9421)& (0.9385)& (0.9457)& (0.9523)& (0.9531)& (0.9549) \\
				60& 1.3029 &0.9975 &1.3441 &0.6292 &0.5536 &0.6723\\
				& (0.9441)& (0.9383)& (0.9457)& (0.9529)& (0.9537)& (0.9525) \\
				70 &1.2344 &0.9652 &1.2958 &0.5849 &0.5586 &0.6421\\
				& (0.9427)& (0.9405)& (0.9453)& (0.9525)& (0.9537)& (0.9544) \\
				80 &1.1892 &0.9032 &1.216 &0.6277 &0.5614 &0.6515\\
				& (0.9398)& (0.9439)& (0.9458)& (0.9538)& (0.9521)& (0.9513) \\
				90 &1.1532 &0.8651 &1.1555 &0.5766 &0.5229 &0.6069 \\
				& (0.9454)& (0.9429)& (0.9451)& (0.9509)& (0.9517)& (0.9526) \\
				100 &1.1215 &0.8227 &1.0825 &0.5687 &0.5021 &0.6022 \\
				& (0.9433)& (0.9416)& (0.9447)& (0.9526)& (0.9532)& (0.9541) \\
				150 &0.9949 &0.6847 &0.8899 &0.5625 &0.4576 &0.5388\\
				& (0.9433)& (0.9415)& (0.9443)& (0.9524)& (0.9529)& (0.9539) \\
				\bottomrule
		\end{tabular}}
	\end{center}
	\vspace{-0.5cm}
\end{table}

\begin{table}[htbp!]
	\begin{center}
		\caption{ALs of $95\%$ asymptotic confidence intervals (ACI) and Bayesian credible intervals (BCI) and their corresponding CPs (in parentheses) with $\alpha=0.3$, $\beta=0.5$, and $\theta=0.2$ with different values of $n$. }
		\label{T5}
		\tabcolsep 4pt
		\small
		\scalebox{0.85}{
			\begin{tabular}{*{11}c*{10}{r@{}l}}
				\toprule
				\multicolumn{1}{c}{} &
				\multicolumn{3}{c}{ACI} & \multicolumn{3}{c}{BCI} \\
				\cmidrule(lr){2-4}\cmidrule(lr){5-7}   
				\multicolumn{1}{c}{$n$}& \multicolumn{1}{c}{$\alpha$} & \multicolumn{1}{c}{$\beta$} & \multicolumn{1}{c}{$\theta$}& \multicolumn{1}{c}{$\alpha$} & \multicolumn{1}{c}{$\beta$} & \multicolumn{1}{c}{$\theta$} \\
				\midrule
				20& 1.2813 &2.0610 &0.6583 &0.6726 &0.8700 &0.4641\\
				& (0.9415)& (0.9389)& (0.9422)& (0.9514)& (0.9532)& (0.9508) \\
				30& 1.0218 &1.4866 &0.6019 &0.6737 &0.8299 &0.4058\\
				& (0.9421)& (0.9406)& (0.9427)& (0.9521)& (0.9519)& (0.9525) \\
				40& 0.9584 &1.3192 &0.5763 &0.6762 &0.8352 &0.3747\\
				& (0.9433)& (0.9396)& (0.9418)& (0.9529)& (0.9517)& (0.9524) \\
				50 & 0.7901 &1.1588 &0.5533 &0.6486 &0.8080 &0.3716\\
				& (0.9428)& (0.9419)& (0.9436)& (0.9537)& (0.9541)& (0.9538) \\
				60 & 0.7628 &0.3015 &0.5425 &0.6259 &0.8386 &0.03270\\
				& (0.9435)& (0.9450)& (0.9441)& (0.9522)& (0.9537)& (0.9515) \\
				70& 0.7443 &0.9798 &0.5283 &0.6411 &0.7595 &0.3166\\
				& (0.9424)& (0.9443)& (0.9449)& (0.9553)& (0.9540)& (0.9521) \\
				80 &0.6848 &0.8818 &0.4665 &0.6312 &0.7588 &0.3139\\
				& (0.9445)& (0.9437)& (0.9456)& (0.9539)& (0.9512)& (0.9496) \\
				90 &0.6621 &0.8065 &0.4317 &0.6257 &0.7613 &0.2987\\
				& (0.9428)& (0.9440)& (0.9425)& (0.9519)& (0.9507)& (0.9515) \\
				100 & 0.6368 &0.7683 &0.4019 &0.5967 &0.7438 &0.2788\\
				& (0.9433)& (0.9417)& (0.9445)& (0.9529)& (0.9513)& (0.9536) \\
				150 &0.4625 &0.5832 &0.3793 &0.1038 &0.0899 &0.2367\\
				& (0.9443)& (0.9429)& (0.9408)& (0.9531)& (0.9547)& (0.9512) \\
				\bottomrule
		\end{tabular}}
	\end{center}
	\vspace{-0.5cm}
\end{table}

\begin{table}[htbp!]
	\begin{center}
		\caption{ALs of $95\%$ asymptotic confidence intervals (ACI) and Bayesian credible intervals (BCI) and their corresponding CPs (in parentheses) with $\alpha=1$, $\beta=0.7$, and $\theta=0.4$ with different values of $n$.}
		\label{T6}
		\tabcolsep 4pt
		\small
		\scalebox{0.85}{
			\begin{tabular}{*{11}c*{10}{r@{}l}}
				\toprule
				\multicolumn{1}{c}{} &
				\multicolumn{3}{c}{ACI} & \multicolumn{3}{c}{BCI} \\
				\cmidrule(lr){2-4}\cmidrule(lr){5-7}   
				\multicolumn{1}{c}{$n$}& \multicolumn{1}{c}{$\alpha$} & \multicolumn{1}{c}{$\beta$} & \multicolumn{1}{c}{$\theta$}& \multicolumn{1}{c}{$\alpha$} & \multicolumn{1}{c}{$\beta$} & \multicolumn{1}{c}{$\theta$} \\
				\midrule
				20 &2.6780 &4.1893 &0.5645 &1.1012 &1.0216 &0.8413 \\
				& (0.9432)& (0.9417)& (0.9455)& (0.9526)& (0.9517)& (0.9529) \\
				30 &2.3077 &3.1644 &0.4515 &1.0359 &0.9473 &0.3928 \\
				& (0.9420)& (0.9423)& (0.9437)& (0.9536)& (0.9524)& (0.9515) \\
				40 &2.1733 &2.6883 &0.3796 &1.0886 &0.9107 &0.3059\\
				& (0.9427)& (0.9429)& (0.9433)& (0.9518)& (0.9525)& (0.9547) \\
				50 &2.0280 &2.4362 &0.3431 &0.9761 &0.9139 &0.2753\\
				& (0.9450)& (0.9439)& (0.9427)& (0.9538)& (0.9516)& (0.9534) \\
				60 &1.9420 &2.2549 &0.3087 &0.9644 &0.9401 &0.2284\\
				& (0.9416)& (0.9428)& (0.9436)& (0.9528)& (0.9529)& (0.9545) \\
				70 &1.7909 &2.0937 &0.2894 &0.9151 &0.8737 &0.2290\\
				& (0.9409)& (0.9434)& (0.9443)& (0.9547)& (0.9538)& (0.9512) \\
				80 &1.7050 &2.0370 &0.2686 &0.8812 &0.9474 &0.2188\\
				& (0.9426)& (0.9432)& (0.9445)& (0.9533)& (0.9504)& (0.9526) \\
				90 &1.6319 &1.9473 &0.2511 &0.8696 &0.8996 &0.2017\\
				& (0.9428)& (0.9439)& (0.9419)& (0.9533)& (0.9528)& (0.9538) \\
				100 &1.5480 &1.8356 &0.2378 &0.8772 &0.8790 &0.1814\\
				& (0.9416)& (0.9445)& (0.9461)& (0.9539)& (0.9528)& (0.9495) \\
				150 &1.2775 &1.6382 &0.1941 &0.7568 &0.8428 &0.1370 \\
				& (0.9425)& (0.9437)& (0.9448)& (0.9516)& (0.9525)& (0.9524) \\
				\bottomrule
		\end{tabular}}
	\end{center}
	\vspace{-0.5cm}
\end{table}

\section{Real data analysis}
Here, two real life data sets are analyzed to illustrate the applicability of the proposed distribution.  \\
\subsection{Flood level data set}
A real life data set related to flood level has been considered to illustrate the applicability of IGLFR distribution. This data set collected by \cite{united1977guidelines} contains $39$ observations of the annual flood discharge rates ($ft^3/s$) of the Floyed river. Flood rates of rivers have a significant socio-economic, political impact and also has an important role in engineering phenomenon. The data set is given below:\\
--------------------------------------------------------------------------------------------------------------------------\\
1460,  4050,  3570,  2060,  1300,  1390,  1720,  6280,  1360,  7440,  5320,  1400,  3240, 2710,  4520,  4840,  8320, 13900, 71500,  6250,  2260,   318,  1330,   970,  1920, 15100, 2870, 20600,  3810,   726,  7500,  7170,  2000,   829, 17300,  4740, 13400,  1940,  5660.\\
--------------------------------------------------------------------------------------------------------------------------\\

\begin{table}[htbp!]
	\begin{center}
		\caption{Goodness-of-fit test for the proposed IGLFR with some other distributions (Dist) models for flood level data.}
		\label{T7}
		\tabcolsep 4pt
		\small
		\scalebox{1}{
			\begin{tabular}{*{11}c*{10}{r@{}l}}
				\\
				\toprule
				\multicolumn{1}{c}{} &
				\multicolumn{3}{c}{Estimates (SE)} & \multicolumn{2}{c}{} \\
				\cmidrule(lr){2-4}  
				\multicolumn{1}{c}{Dist}& \multicolumn{1}{c}{$\alpha$} & \multicolumn{1}{c}{$\beta$} & \multicolumn{1}{c}{$\theta$}& \multicolumn{1}{c}{K-S}& \multicolumn{1}{c}{$p$-value} \\
				\midrule
				IGLFR& 2377.2233& 2.2279& 1.1717& {\bf 0.0851}& {\bf 0.9173} \\
				GIE& 1.1716& 2377.6415& & 0.0876& 0.8999 \\
				GIW& 9.6911& 1.0181& 242.0750&  0.0892& 0.8884 \\
				IW& 2441.9167& 1.0180& & 0.0893& 0.8876 \\
				IPL& 1.0181& 2444.1216& & 0.0894& 0.8871 \\
				IG& 210.5542& 10.1148& & 0.0905& 0.8718 \\
				\bottomrule
			\end{tabular}}
		\end{center}
	\vspace{-0.5cm}
\end{table}
To fit this data set with the proposed IGLFR distribution, we compare it with some well-known inverse distributions such as inverse Weibull (IW), generalized inverse exponential (GIE), inverse power Lindley (IPL), inverse Gompertz (IG) and generalized inverse Weibull (GIW) distributions. The goodness-of-fit test has been compared of the above mentioned distribution models by using Kolmogorov-Smirnov (K-S) distance and corresponding $p$-value. Table $\ref{T7}$ represents the MLEs of the parameters and the corresponding K-S distance and associated $p$-values of the competing model distributions. From Table $\ref{T7}$, it has been observed that the smallest K-S distance and the largest $p$-value has been computed for IGLFR distribution. Figure $6$ contains the empirical CDF (ECDF) plot, probability-probability (P-P) plot, quantile-quantile (Q-Q) plot, boxplot, histogram with density plot and TTT plot for IGLFR distribution under the given real life data. This shows that our proposed IGLFR distribution is a better fit to the given flood level data than the other above mentioned inverse distributions.  The point and interval estimates of the model parameters of IGLFR distribution for the flood level data have been tabulated in Table $\ref{T8}$. From Table $\ref{T8}$, it has been observed that BCIs perform better than ACIs in terms of interval length.   
\begin{table}[htbp!]
	\begin{center}
		\caption{Point and interval estimates of the parameters of IGLFR distribution for flood level data.}
		\label{T8}
		\tabcolsep 4pt
		\small
		\scalebox{1}{
			\begin{tabular}{*{11}c*{10}{r@{}l}}
				\\
				\toprule
				\multicolumn{1}{c}{Parameter}& \multicolumn{1}{c}{MLE} & \multicolumn{1}{c}{ACI} & \multicolumn{1}{c}{Bayes}& \multicolumn{1}{c}{BCI}\\
				\midrule
				$\alpha$& 2377.2233& (1421.6925, 3332.7552)& 2418.4050& (2373.3980, 2461.5000) \\
				$\beta$& 2.2279& (0.4887, 3.9670)& 1.9793& (1.5386, 2.4238) \\
				$\theta$& 1.1717& (0.6660, 1.6775)& 1.2513& (0.8811, 1.6287) \\
			\bottomrule
			\end{tabular}}
		\end{center}
	\vspace{-0.5cm}
\end{table}

	\begin{figure}[htbp!]
	\subfigure[]{\includegraphics[height=1.5in,width= 2.15 in]{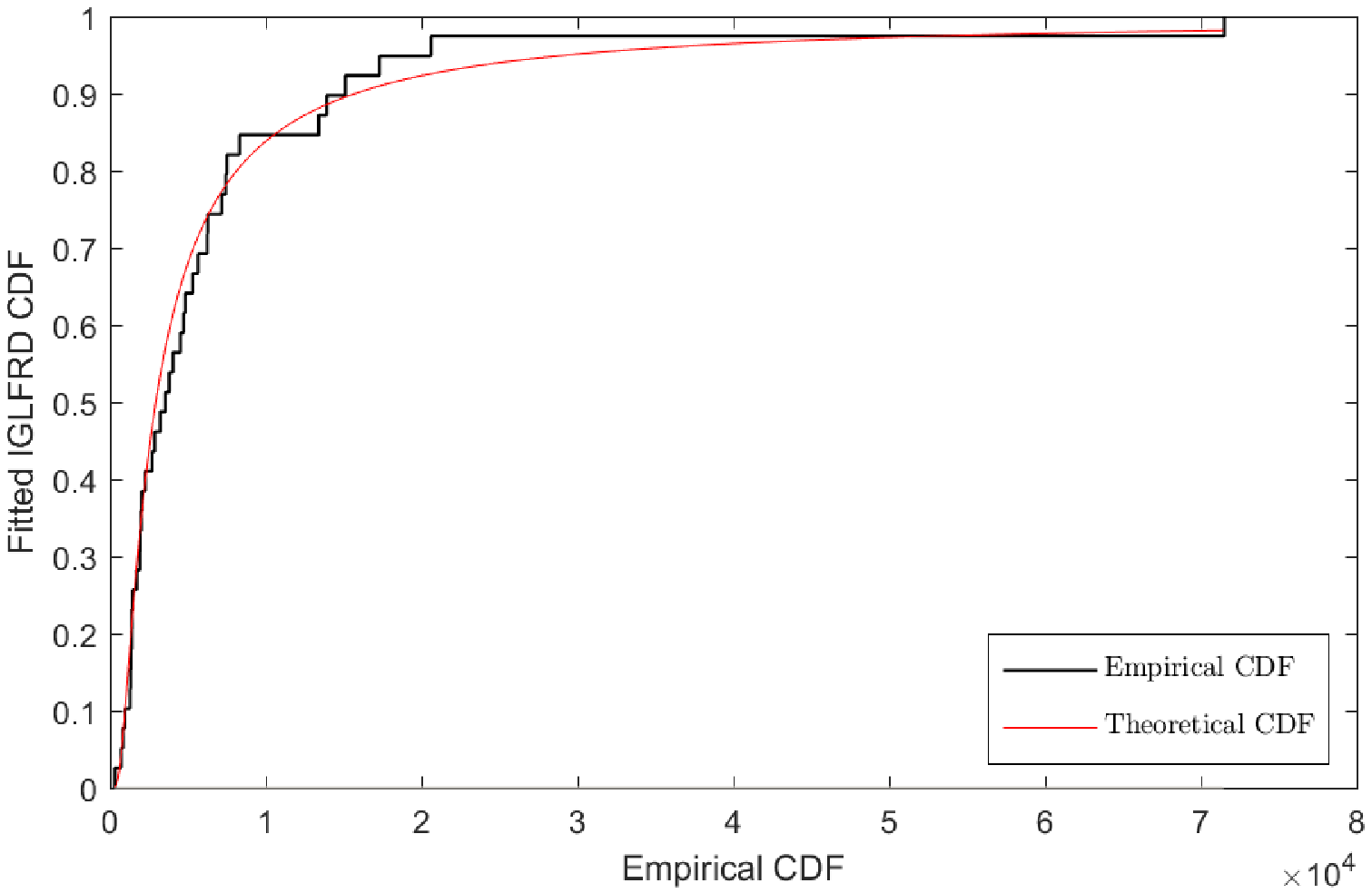}}
	\subfigure[]{\includegraphics[height=1.5in,width=2.15 in]{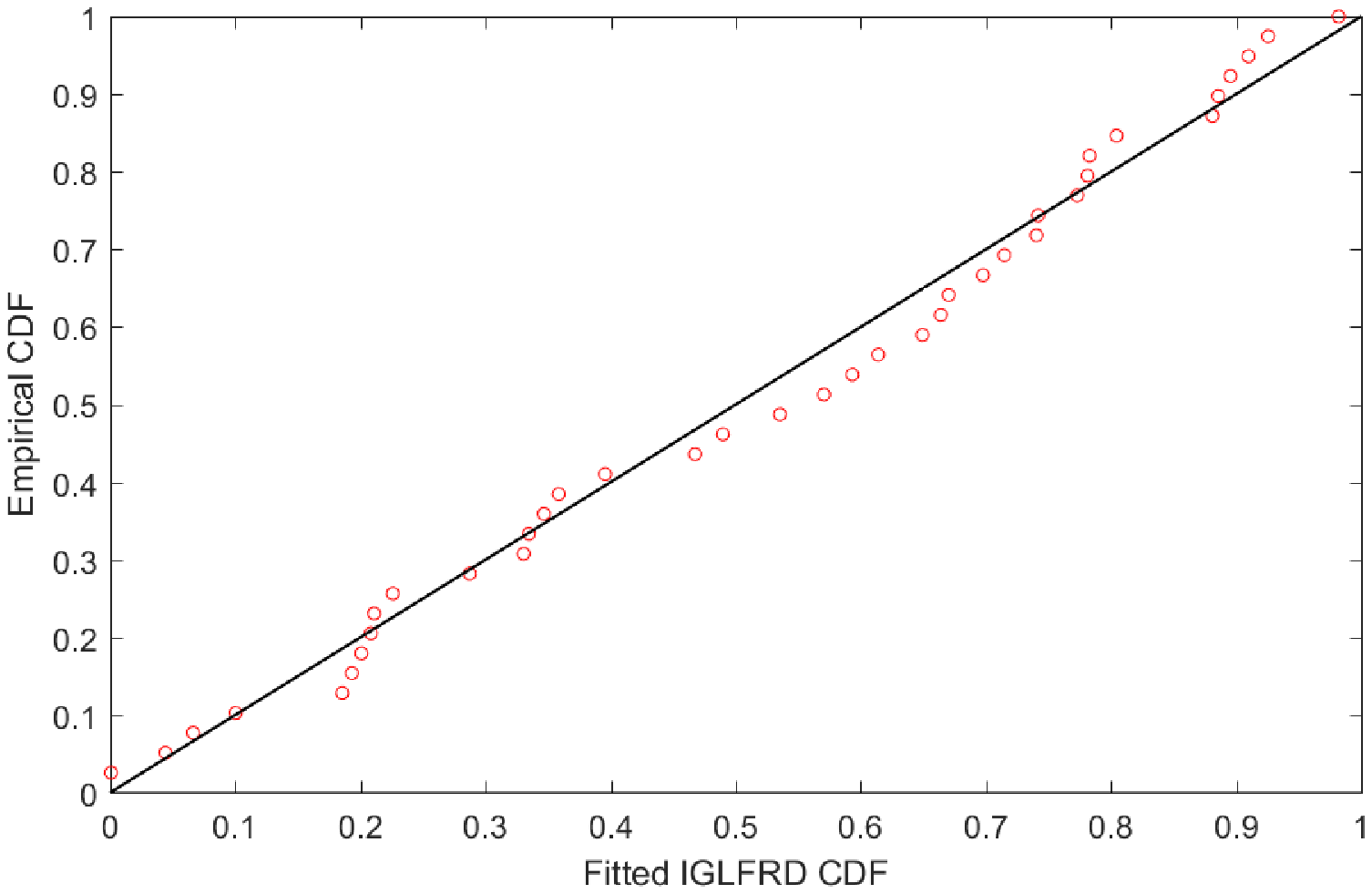}}
	\subfigure[]{\includegraphics[height=1.5in, width=2.15 in]{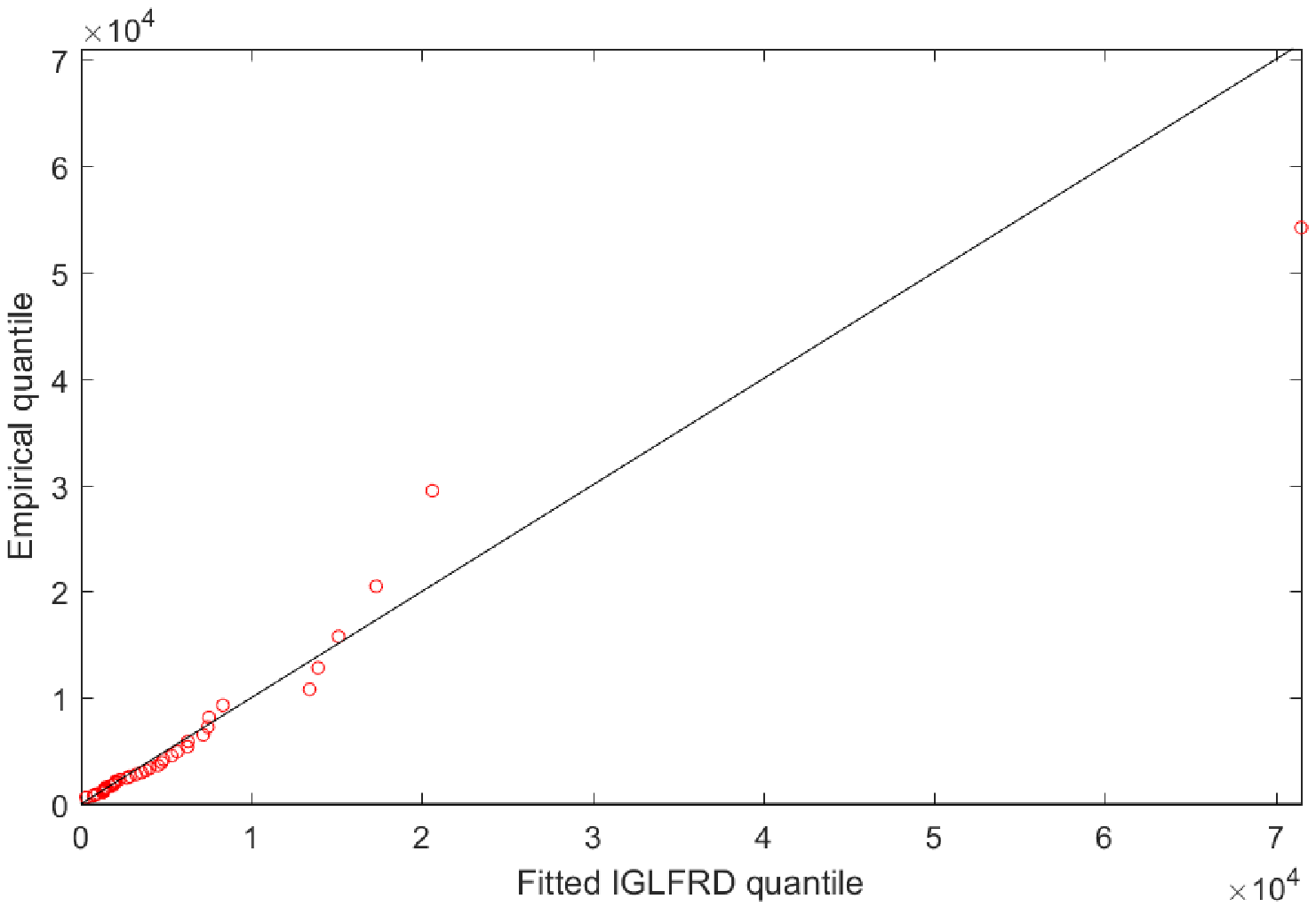}}
    \subfigure[]{\includegraphics[height=1.5in, width=2.15 in]{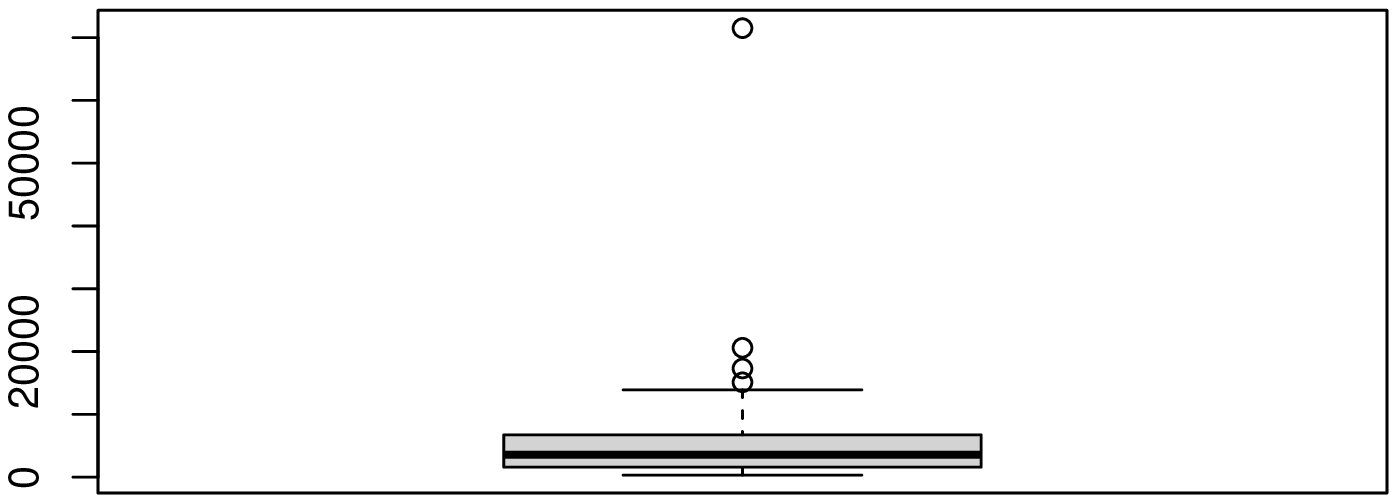}}
    \subfigure[]{\includegraphics[height=1.5in, width=2.15 in]{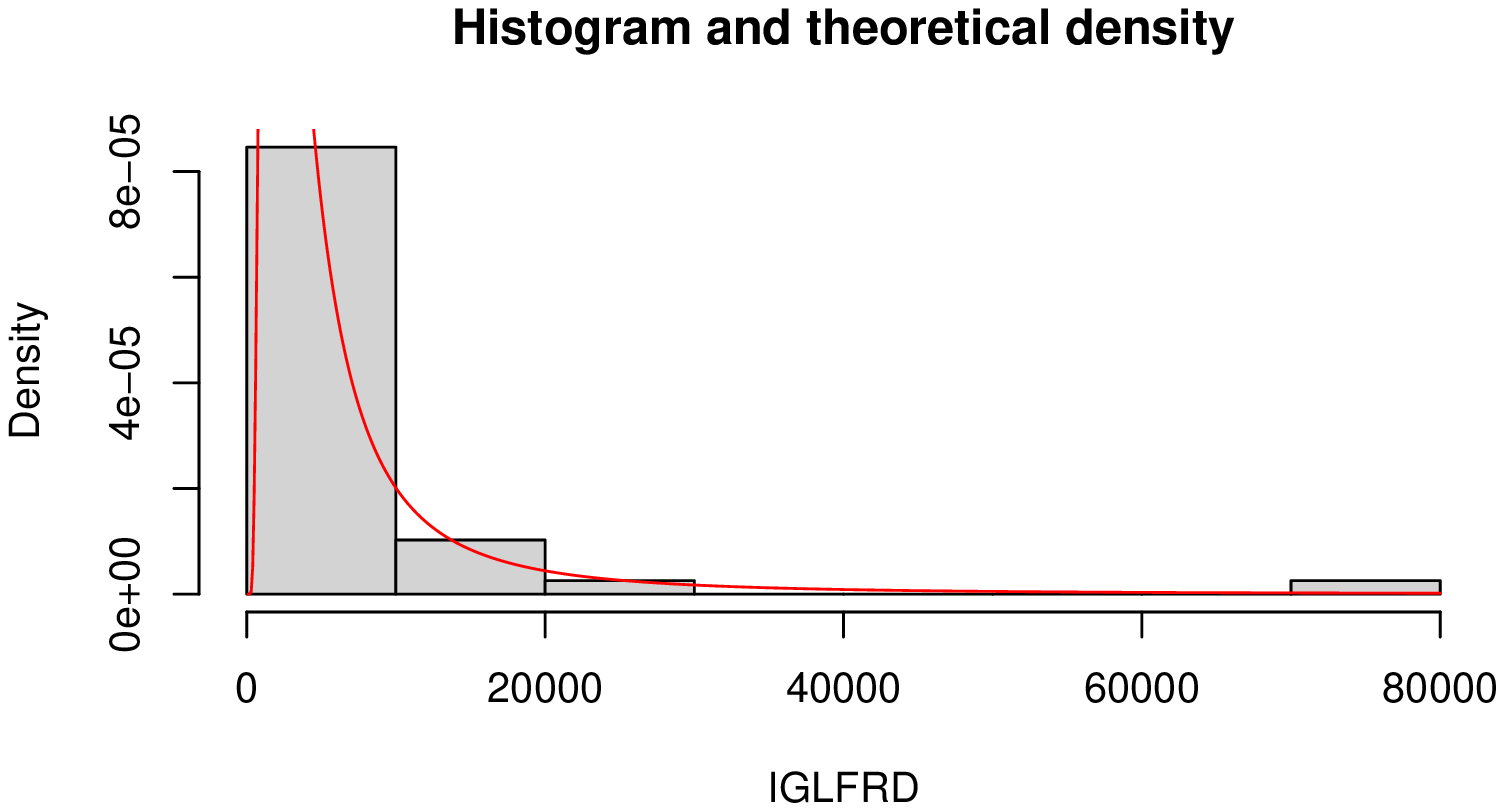}}
    \subfigure[]{\includegraphics[height=1.5in, width=2.15 in]{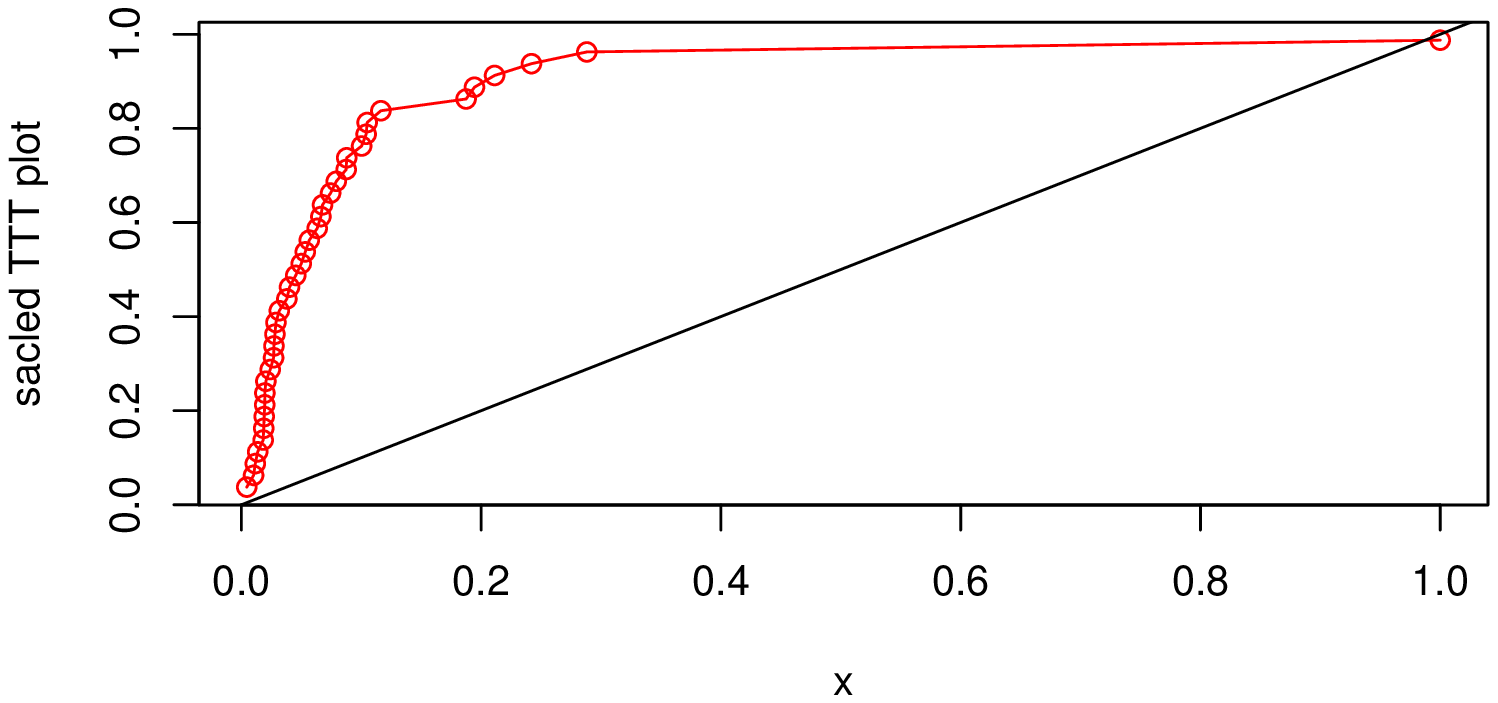}}
	\label{c1}
	\caption{(a) Empirical CDF, (b) P-P plot, (c) Q-Q plot, (d) boxplot, (e) histogram with theoretical density and (f) TTT plot for the IGLFR distribution under the flood data.}
\end{figure}

\subsection{Covid-19 data set}
A Covid-19 data set from \cite{liu2021modeling} is considered to check the applicability of the proposed IGLFR distribution model. This data set contains mortality rates of Covid-19 patients of Canada which is given by\\ 
---------------------------------------------------------------------------------------------------------------------------
1.5157, 1.5806, 1.9048, 2.1901, 2.4141, 2.4946, 2.5261, 2.6029, 2.7704, 2.7957, 2.8349, 2.8636, 2.9078, 3.0914, 3.1091, 3.1091, 3.1444, 3.1348, 3.2110, 3.2135, 3.2218, 3.2823, 3.3592, 3.3769, 3.3825, 3.5146, 3.6346, 3.6426, 3.8594, 4.0480, 4.1685, 4.2202, 4.2781, 4.9274, 4.9378, 6.8686.\\
---------------------------------------------------------------------------------------------------------------------------
\begin{table}[htbp!]
	\begin{center}
		\caption{Goodness-of-fit test for the proposed IGLFR distribution with some other distributions (Dist) models for Covid-19 data.}
		\label{T9}
		\tabcolsep 4pt
		\small
		\scalebox{1}{
			\begin{tabular}{*{11}c*{10}{r@{}l}}
				\\
				\toprule
				\multicolumn{1}{c}{} &
				\multicolumn{3}{c}{Estimates} & \multicolumn{2}{c}{} \\
				\cmidrule(lr){2-4}  
				\multicolumn{1}{c}{Dist}& \multicolumn{1}{c}{$\alpha$} & \multicolumn{1}{c}{$\beta$} & \multicolumn{1}{c}{$\theta$}& \multicolumn{1}{c}{K-S}& \multicolumn{1}{c}{$p$-value} \\
				\midrule
				IGLFR& 11.7507& 1.7358& 30.6509& {\bf 0.1073}& {\bf 0.7614} \\
				GLFR& 0.6272& 0.1843& 12.4830& 0.1081& 0.7544 \\
				IW& 23.4053& 3.1691& & 0.1737& 0.2023 \\
				GIW& 1.7154& 3.1692& 4.2315&  0.1738& 0.2018 \\
				GIR& 1.2174& 2.6005& ~~& 0.2789& 0.0057 \\
				IG&  0.0908& 6.5107& & 0.2006& 0.0961 \\
				\bottomrule
		\end{tabular}}
	\end{center}
	\vspace{-0.5cm}
\end{table}

To fit this data set with the proposed IGLFR distribution, we compare it with some well-known inverse distributions such as generalized linear failure rate (GLFR), inverse Weibull (IW), generalized inverse Rayleigh (GIR), inverse Gompertz (IG) and generalized inverse Weibull (GIW) distributions. The goodness-of-fit test has been compared of the above mentioned distribution models by using Kolmogorov-Smirnov (K-S) distance and corresponding $p$-value. Table $\ref{T9}$ represents the MLEs of the parameters and the corresponding K-S distance and associated $p$-values of the competing model distributions. From Table $\ref{T9}$, the smallest K-S distance and the largest $p$-value yields that IGLFR distribution fits the given covid-19 data better than the other mentioned distributions.  Figure $7$ contains the empirical CDF (ECDF) plot, probability-probability (P-P) plot, quantile-quantile (Q-Q) plot, boxplot, histogram with density plot and TTT plot for IGLFR distribution under the given covid-19 data. This shows that our proposed IGLFR distribution is a better fit to the covid-19 data than the other above mentioned inverse distributions. The point and interval estimates of the model parameters of IGLFR distribution for the covid-19 data have been tabulated in Table $\ref{T10}$. From Table $\ref{T10}$, it has been observed that BCIs perform better than ACIs in terms of interval length. 
	\begin{figure}[htbp!]
	\subfigure[]{\includegraphics[height=1.5in,width= 2.15 in]{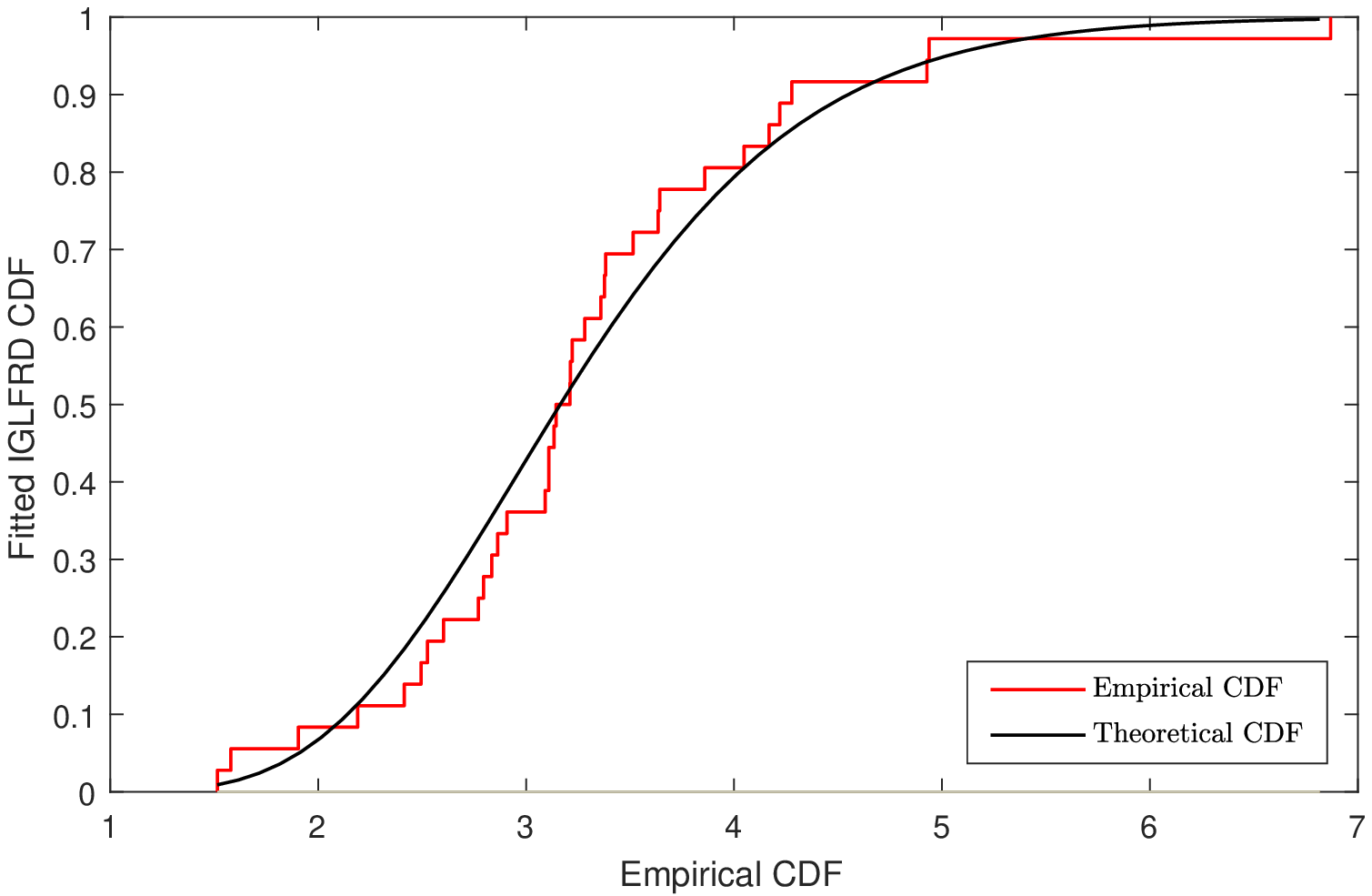}}
	\subfigure[]{\includegraphics[height=1.5in,width=2.15 in]{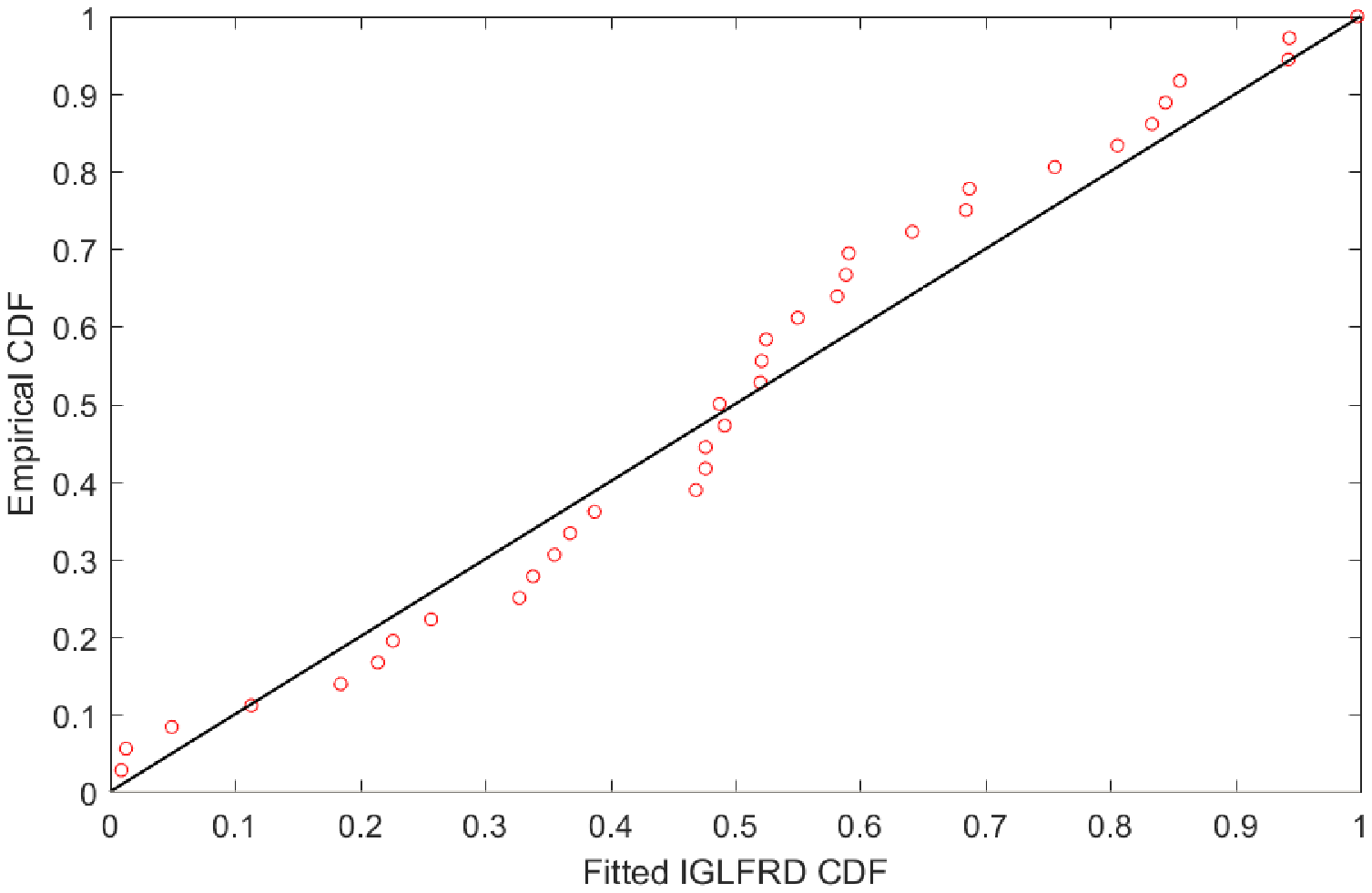}}
	\subfigure[]{\includegraphics[height=1.5in, width=2.15 in]{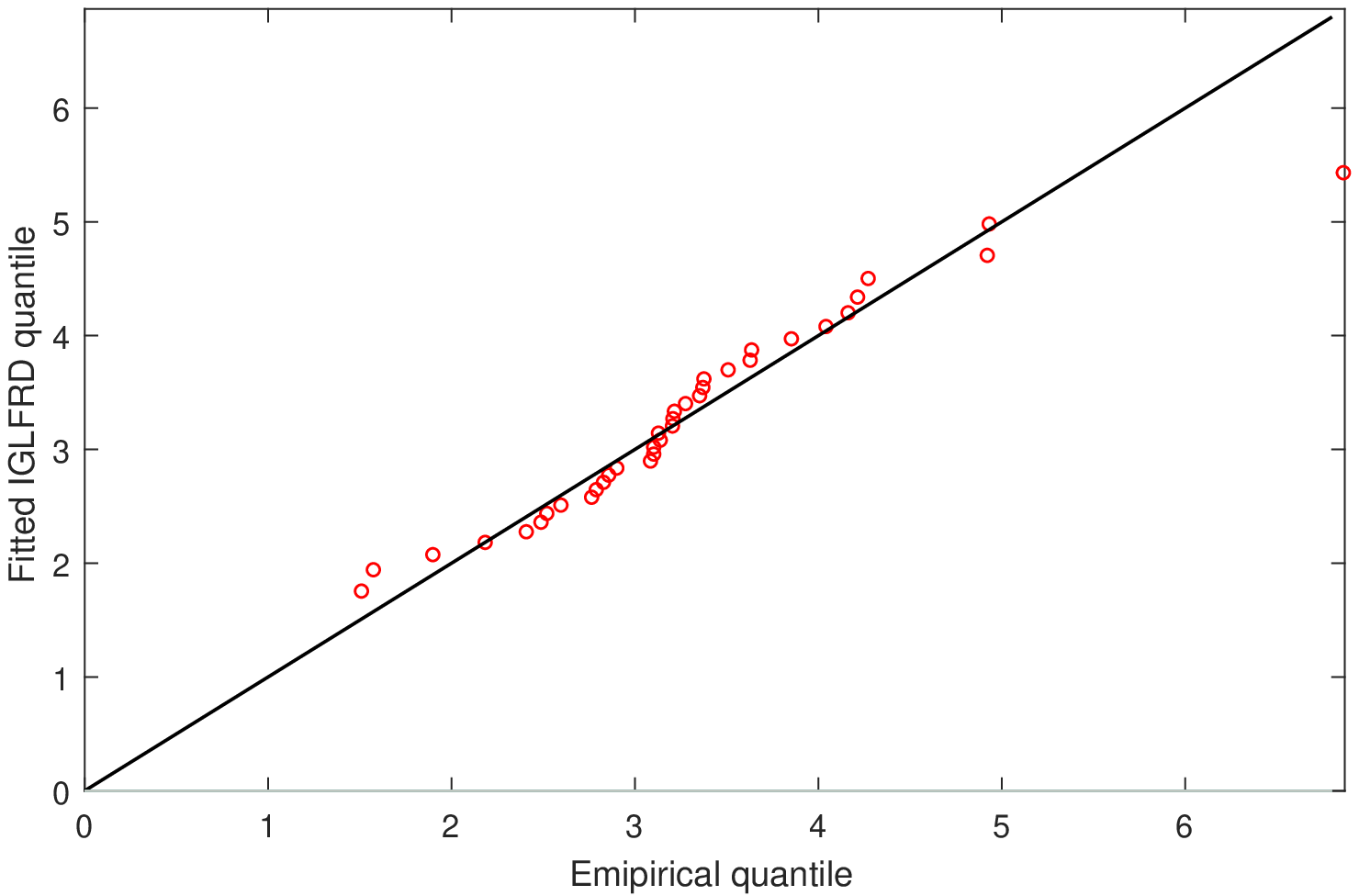}}
	\subfigure[]{\includegraphics[height=1.5in, width=2.15 in]{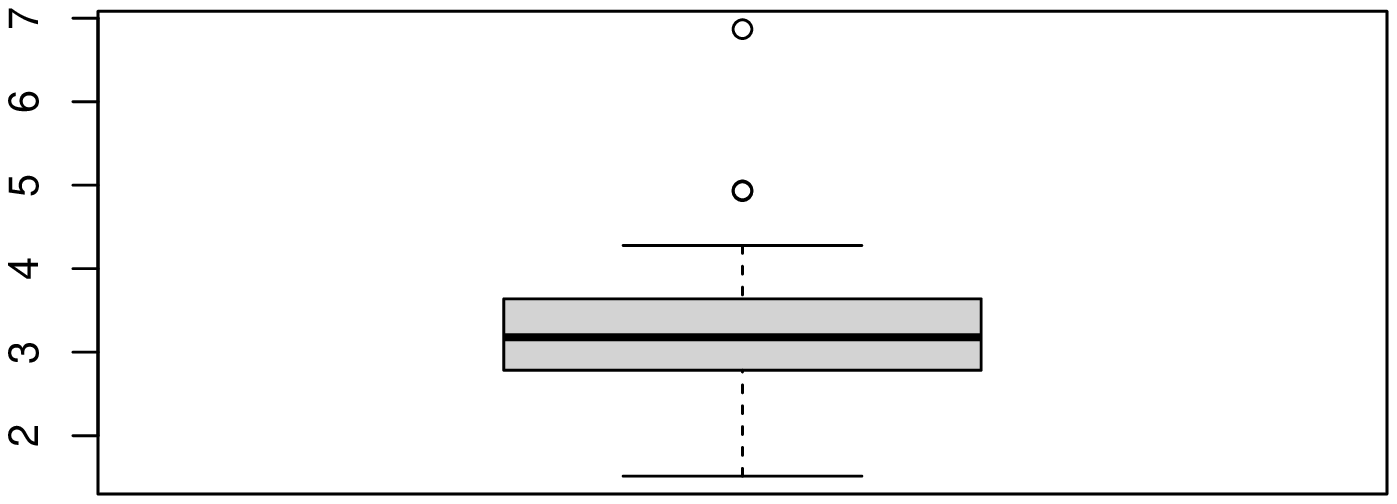}}
	\subfigure[]{\includegraphics[height=1.5in, width=2.15 in]{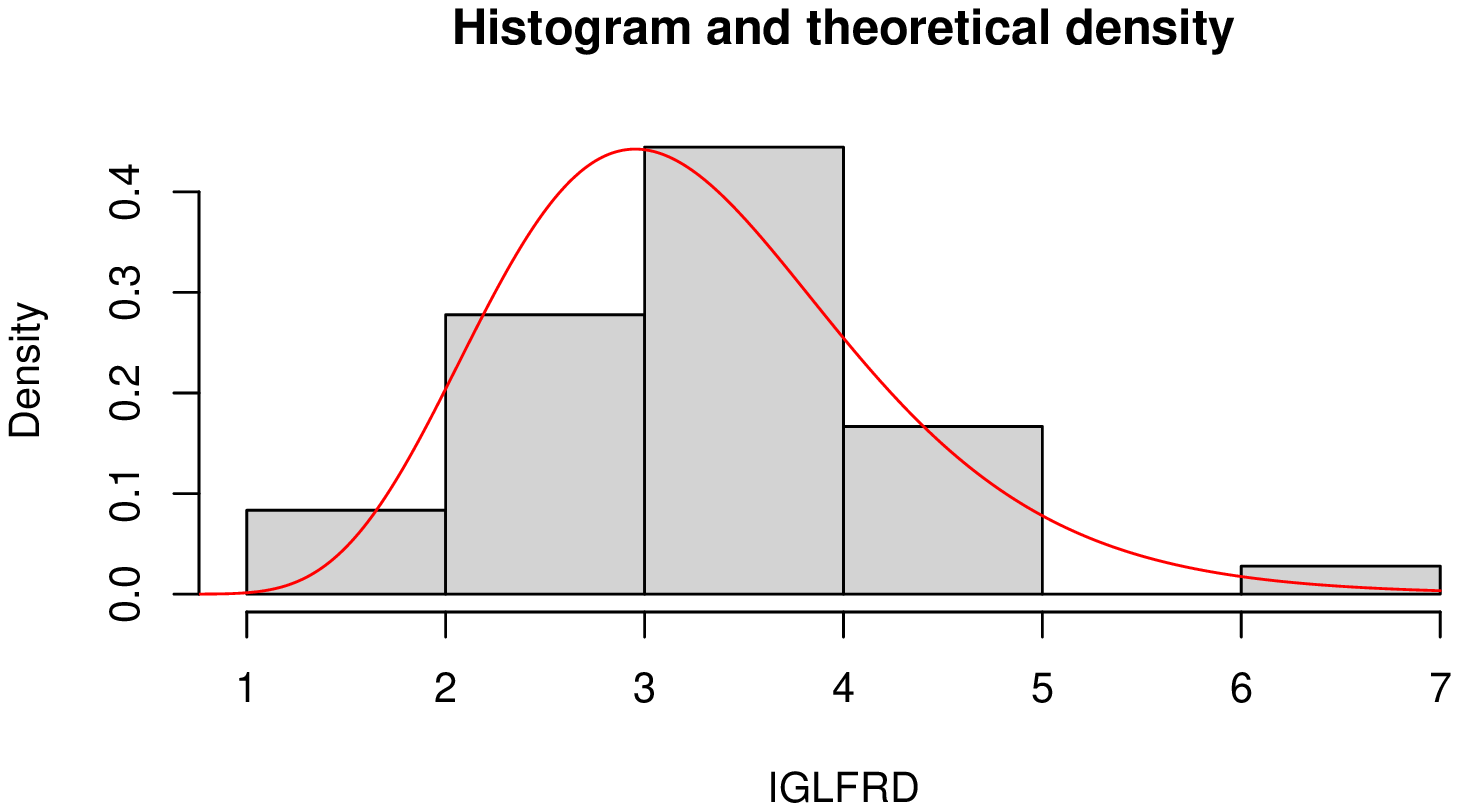}}
	\subfigure[]{\includegraphics[height=1.5in, width=2.15 in]{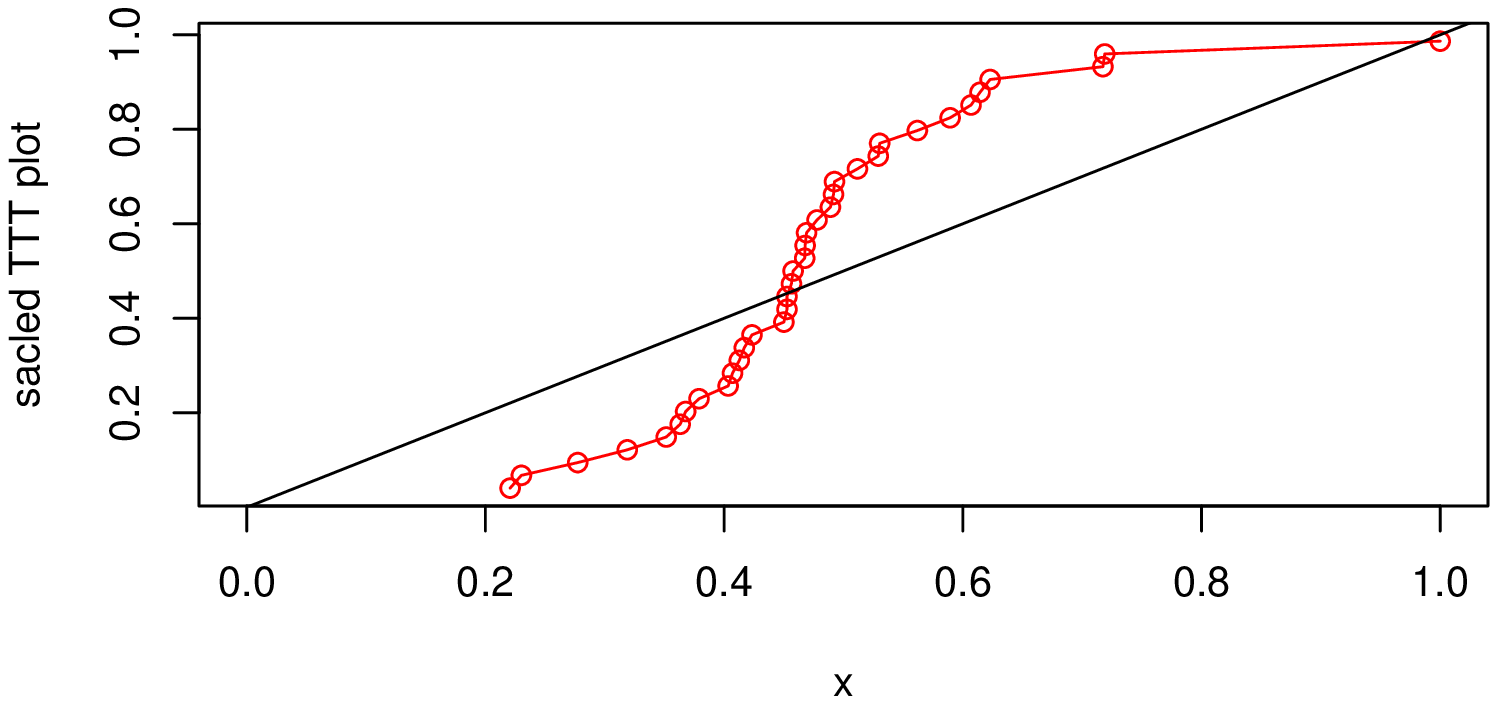}}
	\label{c2}
	\caption{(a) Empirical CDF, (b) P-P, (c) Q-Q plot, (d) boxplot, (e) histogram with theoretical density and (f) TTT plotfor the IGLFR distribution under the covid-19 data.}
\end{figure}

\begin{table}[htbp!]
	\begin{center}
		\caption{Point and interval estimates of the parameters of IGLFR distribution for flood level data.}
		\label{T10}
		\tabcolsep 4pt
		\small
		\scalebox{1}{
			\begin{tabular}{*{11}c*{10}{r@{}l}}
				\\
				\toprule
				\multicolumn{1}{c}{Parameter}& \multicolumn{1}{c}{MLE} & \multicolumn{1}{c}{ACI} & \multicolumn{1}{c}{Bayes}& \multicolumn{1}{c}{BCI}\\
				\midrule
				$\alpha$& 11.7507& (0.1104, 23.3909)& 10.9631& (9.9737, 12.0084) \\
				$\beta$& 1.7315& (0.1053, 3.2439)& 1.3119& (0.7948, 1.8901) \\
				$\theta$& 30.6446& (0.8895, 68.4806)& 24.9705& (20.8091, 30.0026) \\
				\bottomrule
		\end{tabular}}
	\end{center}
	\vspace{-0.5cm}
\end{table}

\section{Conclusions}
In this communication, a new statistical model, called as IGLFR distribution has been introduced along with its  statistical properties. Point as well as interval estimates of the parameters have been obtained. We have obtained MLEs, MPS and Bayes estimates. For Bayes estimates, SELL is used. A detailed simulation study has been performed using various packages in $R$ software to see the performance of the proposed estimates.   From the simulated tables, it has been noticed that Bayesian estimates are more efficient than classical estimates based on average biaes and MSEs for point estimates and ALs and CPs for interval estimates. Finally, two data sets are considered and analyzed. 
\\
\\
\textbf{ Disclosure Statement:} Both the authors state that they do not have any conflict of interest.

\bibliography{sk}
\end{document}